 \documentclass[12pt]{amsart}
\usepackage{a4wide}
\usepackage{amssymb}
\usepackage{amsmath}
\usepackage{amsfonts}
\usepackage{amsthm}
\usepackage{graphicx}
\usepackage{epsfig}
\usepackage{longtable}
\newtheorem{thm}{Theorem}%[section]
\newtheorem{lem}[thm]{Lemma}
\newtheorem{cor}[thm]{Corollary}
\newtheorem{conj}{Conjecture}
\newtheorem{prop}{Proposition}

\theoremstyle{definition}

\theoremstyle{remark}
\newtheorem*{rmk}{Remark}

\newcommand{\eps}{\varepsilon}

\newcommand{\DEF}{{:=}}

\newcommand{\Cset}{\mathbb{C}}
\newcommand{\Nset}{\mathbb{N}}

\newcommand{\Rset}{\mathbb{R}}

\newcommand{\Sp}{\mathbb{S}}

\newcommand{\PT}[1]{\mathbf{#1}}

\newcommand{\re}{\mathop{\mathrm{Re}}}

\DeclareMathOperator{\dd}{\mathrm{d}}

\DeclareMathOperator{\CAP}{cap}

\DeclareMathOperator{\DirichletL}{L}
\DeclareMathOperator{\logEnergy}{E_{\mathrm{log}}}
\DeclareMathOperator{\gammafcn}{\Gamma}

\DeclareMathOperator{\digammafcn}{\psi}

\DeclareMathOperator{\Energy}{E}

\DeclareMathOperator{\res}{Res}

\DeclareMathOperator{\zetafcn}{\zeta}

\DeclareMathOperator{\HyperF}{F}
\DeclareMathOperator{\HyperTildeF}{\tilde{F}}
\newcommand{\Hypergeom}[5]{{\sideset{_#1}{_#2}\HyperF\!\left(\substack{\displaystyle#3\\\displaystyle#4};#5\right)}}
\newcommand{\HypergeomReg}[5]{{\sideset{_#1}{_#2}\HyperTildeF\!\left(\substack{\displaystyle#3\\\displaystyle#4};#5\right)}}

\newcommand{\Pochhsymb}[2]{{\left(#1\right)_{#2}}}

\allowdisplaybreaks[1]

\title[The next-order term]{The next-order term for optimal Riesz and logarithmic energy asymptotics on the   sphere} 
\author{ J. S. Brauchart\textasteriskcentered, D. P. Hardin\textdagger, and E. B. Saff\textdagger } %\textasteriskcentered
\thanks{\noindent \textasteriskcentered Recipient of an {\sc APART}-fellowship of the Austrian Academy of Sciences. Research conducted at the Center for Constructive Approximation at Vanderbilt University and the School of Mathematics and Statistics at the University of New South Wales. \\
\textdagger The research of this author was supported, in part,
by the U. S. National Science Foundation under grants DMS-0808093 and DMS-1109266.  
}

\dedicatory{Dedicated to Paco Marcell\'{a}n on the occasion of his 60-th birthday.}

\date{\today}

\begin{document}

\address{J. S. Brauchart:
School of Mathematics and Statistics, 
University of New South Wales, 
Sydney, NSW, 2052, Australia}

\email{j.brauchart@unsw.edu.au}
\address{D. P. Hardin and E. B. Saff:
Center for Constructive Approximation, 
Department of Mathematics, 
Vanderbilt University, 
Nashville, TN 37240,  
USA }
\email{Doug.Hardin@Vanderbilt.Edu}
\email{Edward.B.Saff@Vanderbilt.Edu}


\begin{abstract} 
We survey known results and present   estimates and conjectures for the next-order term in the asymptotics of the optimal logarithmic energy and Riesz $s$-energy of $N$ points on the unit sphere in $\mathbb{R}^{d+1}$, $d\geq 1$.  The conjectures are based on analytic continuation assumptions (with respect to $s$) for the coefficients in the asymptotic expansion (as $N\to \infty$) of the optimal $s$-energy. 
\end{abstract}

\maketitle

\section{Introduction}

Let $\mathbb{S}^d$ denote the unit sphere in the Euclidean space $\mathbb{R}^{d+1}$, $d\geq 1$. The discrete logarithmic energy problem on $\mathbb{S}^d$ is concerned with investigating the properties of $N$-point systems $\PT{x}_1^*, \dots, \PT{x}_N^*$ on $\mathbb{S}^d$ maximizing the product of all mutual pairwise Euclidean distances 
\begin{equation*}
%\label{product}
M( \PT{x}_1, \dots, \PT{x}_N ) \DEF \prod_{j\neq k} \left| \PT{x}_{j} - \PT{x}_{k} \right|
=\prod_{1\le j< k\le N} \left| \PT{x}_{j} - \PT{x}_{k}\right|^2,
\end{equation*}
or equivalently, minimizing 
\begin{equation} \label{log.energy}
\logEnergy( \PT{x}_1, \dots, \PT{x}_N ) \DEF - \log M( \PT{x}_1, \dots, \PT{x}_N ) = \sum_{j \neq k} \log \frac{1}{\left| \PT{x}_j - \PT{x}_k \right|}=2\sum_{1\le j< k\le N} \log \frac{1}{\left| \PT{x}_j - \PT{x}_k \right|}
% \DEF \mathop{\sum_{j=1}^N \sum_{k=1}^N}_{j \neq k} \log \frac{1}{\left| \PT{x}_j - \PT{x}_k \right|}
\end{equation}
over all $N$-point configurations  $(\PT{x}_1, \dots, \PT{x}_N)$ on $\mathbb{S}^d$.

One goal of this article is to investigate the asymptotic expansion as $N$ goes to infinity of the {\em minimum ($N$-point) logarithmic energy}  of $\Sp^{d}$
\begin{equation*}
\mathcal{E}_{\mathrm{log}}(\mathbb{S}^d; N) \DEF \min \left\{ \logEnergy( \PT{x}_1, \dots, \PT{x}_N ) \mid \PT{x}_1, \dots, \PT{x}_N \in \mathbb{S}^d \right\} = \logEnergy( \PT{x}_1^*, \dots, \PT{x}_N^* ), \quad N \geq 2.
\end{equation*}
We remark that in his list of ``Mathematical problems for the next century'' Smale~\cite{Sm1998,Sm2000} posed as {\em Problem \#7} the challenge to design a fast (polynomial time) algorithm for generating ``nearly'' optimal logarithmic energy points on the unit sphere in $\mathbb{R}^3$ that satisfy
\begin{equation} \label{smales.problem}
\logEnergy(\PT{x}_1,\ldots,\PT{x}_N) - \mathcal{E}_{\mathrm{log}}(\mathbb{S}^2;N) \leq c \log N \qquad \text{for some universal constant $c$.}
\end{equation}
This problem emerged from computational complexity theory (cf. Shub and Smale~\cite{ShSm1993}). 

The right-hand side of \eqref{log.energy} is referred to  as the {\em discrete logarithmic energy} of the normalized counting measure $\mu[\PT{x}_1, \dots, \PT{x}_N]$, which places the  charge $1/N$ at each point $\PT{x}_1, \dots, \PT{x}_N$. By the {\em continuous logarithmic energy} of a (Borel) probability measure $\mu$ supported on $\Sp^{d}$ we mean
\begin{equation*}
%\label{energy.of.measure}
\mathcal{I}_{\mathrm{log}}[\mu] \DEF \iint  \log\frac{1}{\left| \PT{x} - \PT{y} \right|} \dd \mu(\PT{x}) \dd \mu(\PT{y}).
\end{equation*}
Note that $ \mathcal{I}_{\mathrm{log}}[\mu]=+\infty$ for any discrete measure $\mu$.
Classical potential theory yields that   $\mathcal{I}_{\mathrm{log}}[\mu]$ is uniquely minimized by the {\em normalized surface area measure $\sigma_d$}, $\int_{\Sp^{d}}\dd\sigma_d=1$, over the class $\mathcal{M}(\mathbb{S}^d)$ of all (Borel) probability measures $\mu$ supported on $\Sp^{d}$; that is, 
\begin{equation} \label{V.log}
V_{\mathrm{log}}(\mathbb{S}^d) \DEF \inf \left\{ \mathcal{I}_{\mathrm{log}}[\mu] \mid \mu \in \mathcal{M}(\mathbb{S}^d) \right\} = \mathcal{I}_{\mathrm{log}}[\sigma_d] = \log \frac{1}{2} + \frac{1}{2} \left[ \digammafcn( d ) - \digammafcn( d / 2 ) \right],
\end{equation}
where $\digammafcn(z) \DEF \gammafcn^\prime(z) / \gammafcn(z)$ is the digamma function. The last expression in \eqref{V.log} follows, for example, from \cite[Eq.~(2.26)]{Br2008}. 

It is  known that the minimum $N$-point logarithmic energy of $\Sp^{d}$ satisfies
\begin{equation}
\label{wagner}
V_{\mathrm{log}}(\mathbb{S}^d) - \frac{1}{2} \frac{\log N}{N} + \frac{c_{1}}{N} \leq \frac{\mathcal{E}_{\mathrm{log}}(\mathbb{S}^d; N)}{N^{2}} \leq V_{\mathrm{log}}(\mathbb{S}^d) - \frac{1}{d} \frac{\log N}{N} + \frac{c_{2}}{N}
\end{equation}
for some constants $c_{1}$, $c_{2}$ depending on $d$ only. The lower bound follows from \cite{Wa1992}. The upper bound follows from an averaging argument as in the proof of Theorem~1 in  \cite{KuSa1998}, which is based on equal area partitions \cite{Le2006}. These bounds give the correct form of the second-order term in the asymptotics of the minimum $N$-point logarithmic energy of $\Sp^{2}$. It should be mentioned that \cite{RaSaZh1994} also cites \cite[p.~150]{La1988} as a reference to a more general method to obtain lower bounds for minimum logarithmic energy valid over more general Riemann surfaces, which in the case of the $2$-sphere also gives the lower bound in \eqref{wagner} (cf. \cite[Sec.~3]{RaSaZh1994}).
% \begin{equation*}
% \mathcal{E}_{\mathrm{log}}(\mathbb{S}^2; N) \geq V_{\mathrm{log}}(\mathbb{S}^2) \, N^2 - \frac{1}{2} N \log N + \mathcal{O}(N), \qquad \text{as $N \to \infty$.}
% \end{equation*}
In a recent paper the first author improved the lower bound for higher-dimensional spheres. In \cite[Lemma 1.1]{Br2008} it was shown that
\begin{equation} \label{brau}
\frac{\mathcal{E}_{\mathrm{log}}(\mathbb{S}^d; N)}{N^{2}} \geq V_{\mathrm{log}}(\mathbb{S}^d) - \frac{1}{d} \frac{\log N}{N} - \frac{C_{d}^{\prime}}{N} + \mathcal{O}_{\eps}(N^{-1-2\eps/d}), \quad N\to\infty,
\end{equation}
where the positive constant $C_{d}^{\prime}$ does not depend on $N$ and is given by
\begin{equation*} % \label{C:d}
C_{d}^{\prime} := V_{\mathrm{log}}(\mathbb{S}^d) + \frac{1}{d} \frac{\gammafcn(d)\gammafcn(1+\lfloor d/2 \rfloor-d/2)}{2^{d}\gammafcn(d/2) \gammafcn(1+\lfloor d/2 \rfloor)} + \frac{1}{2} \sum_{r=1}^{\lfloor d/2\rfloor} \frac{1}{r} > 0;
\end{equation*}
the number $\eps$ satisfies $0<\eps \, <1$ ($d$ even) or $0<\eps \, <1/2$ ($d$ odd). Combining the upper bound in \eqref{wagner} and the lower bound in \eqref{brau} we obtain the following asymptotic expansion of the minimum $N$-point logarithmic energy of $\Sp^{d}$:
\begin{equation}\label{logas1}
\mathcal{E}_{\mathrm{log}}(\mathbb{S}^d; N) = V_{\mathrm{log}}(\mathbb{S}^d) \, N^2 - \frac{1}{d} N \log N + \mathcal{O}(N), \qquad N\to\infty.
\end{equation}
We remark that the minimum $N$-point logarithmic energy of the unit circle $\Sp$ is attained at the $N$-th roots of unity and (cf., for example, \cite{BrHaSa2009}) 
\begin{equation*}
\mathcal{E}_{\mathrm{log}}(\mathbb{S}; N) = - N \log N, \qquad N \geq 2.
\end{equation*}
(In this case $V_{\mathrm{log}}(\mathbb{S}) = 0$.)

In view of \eqref{logas1} it is tempting to ask if the following limit exists:
\begin{equation} \label{lim.conj}
\lim_{N \to \infty} \frac{1}{N} \left[ \mathcal{E}_{\mathrm{log}}(\mathbb{S}^d; N) - V_{\mathrm{log}}(\mathbb{S}^d) \, N^2 + \frac{1}{d} N \log N \right] = \ ?
\end{equation}
In particular, it is plausible and consistent with the lower bound \eqref{brau}     that this limit is negative for all $d\geq2$ if it exists. Indeed, in \cite{RaSaZh1994, RaSaZh1995}
it is shown (also see below) that on the unit sphere in $\mathbb{R}^3$ (that is, $d=2$) the square-bracketed expression in \eqref{lim.conj} can be bounded by negative quantities from below and above.

In formulating our Conjecture~\ref{conj:log.2.sphere} for $\mathcal{E}_{\log}({\mathbb S}^2; N)$   (see Section~\ref{section4}), we will make heavy use of the observation that the discrete logarithmic energy is the limiting case (as $s\to0$) of the {\em Riesz $s$-energy}
\begin{equation*}
\Energy_{s}( \PT{x}_1, \dots, \PT{x}_N ) \DEF \sum_{j\neq k} \frac{1}{\left| \PT{x}_{j} - \PT{x}_{k} \right|^{s}}=2\sum_{1\le j< k\le N}   \frac{1}{\left| \PT{x}_j - \PT{x}_k \right|^s},
\end{equation*}
by means of $(1/| \PT{\cdot} |^{s}-1)/s\to\log(1/| \PT{\cdot} |)$ as $s\to0$.  For $N\ge 2$, we shall also consider  the {\em optimal $N$-point Riesz $s$-energy} of a compact set $A\subset \Rset^{d+1}$  defined by
\begin{equation}\label{minEsDef}
\mathcal{E}_{s}(A; N) \DEF \begin{cases}\min \left\{ \Energy_s( \PT{x}_1, \dots, \PT{x}_N ) \mid \PT{x}_1, \dots, \PT{x}_N \in A \right\} & \text{for $s\ge 0$ or $s=\log$},\\
\max \left\{\Energy_s( \PT{x}_1, \dots, \PT{x}_N ) \mid \PT{x}_1, \dots, \PT{x}_N \in A \right\} & \text{for $s<0$},
\end{cases}
\end{equation}   where we note that $\mathcal{E}_0(A; N) = N^2 - N$, which is attained by any $N$-point configuration on $A$.  Furthermore, it is known (cf. \cite{MEBook}) that 
\begin{equation}\label{dEsElog}
\left. \frac{\dd}{\dd s} \mathcal{E}_s(A,N)\right|_{s=0^+}=\mathcal{E}_{\log}(A,N) \qquad (N\ge 2).
\end{equation}

As in the logarithmic case, the {\em optimal continuous $s$-energy} of $\mathbb{S}^d$ for $-2 < s < d$ is given by
\begin{equation} \label{V.s.d}
V_s(\mathbb{S}^d) \DEF \iint \frac{1}{|\PT{x}-\PT{y}|^s}\dd \sigma_d(\PT{x}) \dd \sigma_d(\PT{y})= 2^{d-s-1} \frac{\gammafcn((d+1)/2) \gammafcn((d-s)/2)}{\sqrt{\pi}\gammafcn(d-s/2)},
\end{equation}
and the $s$-capacity of $\mathbb{S}^d$ is given by $\CAP_s( \mathbb{S}^d )\DEF 1/V_s(\mathbb{S}^d)$.

In Sections~\ref{section3} and \ref{section4} we will discuss estimates and asymptotics for the optimal $N$-point $s$-energy on ${\mathbb S}^d$.  

\section{The logarithmic energy on the unit sphere in $\mathbb{R}^3$}

Let $\Delta_{\mathrm{log}}(N) \, N$ be the remainder term in
\begin{equation*}
\mathcal{E}_{\mathrm{log}}(\mathbb{S}^2; N) = V_{\mathrm{log}}(\mathbb{S}^2) \, N^2 - \frac{1}{2} N \log N + \Delta_{\mathrm{log}}(N) \, N.
\end{equation*}
Rakhmanov, Saff and Zhou~\cite{RaSaZh1994, RaSaZh1995} proved the following estimates\footnote{We remark that there is a typographical error in the sign of the second ratio in the upper bound  in \cite[Eq.~(3.10)]{RaSaZh1994}.}
\begin{align}
\liminf_{N \to \infty} \Delta_{\mathrm{log}}(N) &\geq - \frac{1}{2} \log \left[ \frac{\pi}{2} \left( 1 - e^{-a} \right)^b \right] = - 0.22553754\dots, \label{liminfbound} \\
\limsup_{N \to \infty} \Delta_{\mathrm{log}}(N) &\leq - \frac{1}{2} \log \frac{\pi \sqrt{3}}{2} + \frac{\pi}{4 \sqrt{3}} = - 0.0469945\dots,  \label{limsupbound}
\end{align}
where
\begin{equation*}
a \DEF \frac{2 \sqrt{2 \pi}}{\sqrt{27}} \left( \sqrt{2 \pi + \sqrt{27}} + \sqrt{2 \pi} \right), \qquad b \DEF \frac{\sqrt{2 \pi + \sqrt{27}} - \sqrt{2\pi}}{\sqrt{2 \pi + \sqrt{27}} + \sqrt{2\pi}},
\end{equation*}
and they stated the following conjecture: 
\begin{conj}[\cite{RaSaZh1994}]
There exist   constants   $C_{\log,2}$ and $D_{\log,2}$, independent of $N$, such that 
\begin{equation} \label{conj.rakhmanov.saff.zhou}
\mathcal{E}_{\mathrm{log}}(\mathbb{S}^2; N) = \left( \frac{1}{2} - \log 2 \right) N^2 - \frac{1}{2} N \log N + C_{\mathrm{log},2} \, N + D_{\mathrm{log},2} \, \log N + \mathcal{O}(1)  \quad \text{as $N \to \infty$.}
\end{equation}
\end{conj}
The authors  of \cite{RaSaZh1994} supported their conjecture by fitting the conjectured formula to the data obtained by high-precision computer experiments to find extremal configurations ($N \leq 200$) and their logarithmic energy by minimizing the absolute $\ell_1$-deviation. This leads to the following approximation\footnote{Note that the logarithmic energy in \cite{RaSaZh1994, RaSaZh1995} is precisely half of ours.}
\begin{equation*}
\mathcal{E}_{\mathrm{log}}( \mathbb{S}^2; N ) \approx \left( \frac{1}{2} - \log 2 \right) N^2 - \frac{1}{2} N \log N - 0.052844 \, N + 0.27644.
\end{equation*}
(In   \cite{RaSaZh1994} it is remarked that the logarithmic term was ignored in the fitting algorithm because it did not appear to be significant.) 

So far no conjecture for the precise value of $C_{\mathrm{log},2}$ in \eqref{conj.rakhmanov.saff.zhou} has been given.  Conjecture~\ref{conj:log.2.sphere} in Section~\ref{section4}   fills this gap. 
Our approach for arriving at this conjecture is as follows: 
Since the logarithmic energy of an $N$-point configuration $X_N$ can be obtained by taking the derivative with respect to $s$ of the Riesz $s$-energy of $X_N$ and letting $s$ go to zero, we shall differentiate the asymptotic expansion of the {\em minimal $N$-point Riesz $s$-energy}   of $\mathbb{S}^2$ (see \eqref{minEsDef}) and then let $s$ go to zero from the right to derive plausible asymptotic formulas for the $N$-point logarithmic energy of $\mathbb{S}^2$.   In this derivation there is an implicit interchange of differentiation and minimization which is formally justified by \eqref{dEsElog}.   

In order to illustrate this approach, we begin with the simplest case, which is the unit circle $\mathbb{S}$.
Recently, we obtained in \cite{BrHaSa2009} the complete asymptotic  expansion of the Riesz $s$-energy $\mathcal{L}_s(N)$ of $N$-th roots of unity, which are the only (up to rotation) optimal energy points on the unit circle if $s>-2$, $s\neq0$, as well as for the logarithmic case. In fact, the expansion is valid for $s \in \Cset$ with $s$ not zero or an odd positive integer:
\begin{equation} \label{circle.asympt}
\begin{split}
\mathcal{L}_s(N) 
&= V_s(\mathbb{S}) \, N^2 + \frac{2\zetafcn(s)}{(2\pi)^s} \, N^{1+s} + \frac{2}{(2\pi)^s} \sum_{n=1}^p \alpha_n(s) \zetafcn(s-2n) N^{1+s-2n} \\
&\phantom{=\pm}+ \mathcal{O}_{s,p}(N^{-1+\re s-2p}) \qquad \text{as $N \to \infty$ ($p = 1, 2, 3,   \dots$),}
\end{split}
\end{equation}
where $\zeta(s)$ is the classical Riemann zeta function. 
The coefficients $\alpha_n(s)$, $n\geq0$, satisfy the generating function relation
\begin{equation*} \label{sinc.power.0}
\left( \frac{\sin \pi z}{\pi z} \right)^{-s} = \sum_{n=0}^\infty \alpha_n(s) z^{2n},
\quad |z|<1, \ s\in \Cset, \qquad \alpha_n(s) = \frac{(-1)^n B_{2n}^{(s)}(s/2)}{(2n)!} \left( 2 \pi \right)^{2n}.
\end{equation*}
Here, $B_n^{(\alpha)}(x)$ denotes the so-called generalized Bernoulli polynomial of degree $n$. In particular, one has the following recurrence relation for $n\ge 1$ (cf. Luke~\cite[p.~34f]{Lu1969I})
\begin{equation} \label{gen.BernoulliP.recc}
B_{2n}^{(2\rho)}(\rho) = - 2 \rho \sum_{m=0}^{n-1} \binom{2n-1}{2m+1} \frac{B_{2m+2}}{2m+2} B_{2n-2-2m}^{(2\rho)}(\rho), \qquad B_0^{(2\rho)}(\rho) = 1,
\end{equation}
where $B_0 = 1$, $B_1 = - 1 / 2$, $B_2 = 1 / 6$, $B_3 = 0$, \dots{} are the Bernoulli numbers.
The constant $V_s(\mathbb{S})$ is the Riesz $s$-energy of the unit circle if $-2 < s < 1$, $s \neq 0$, and its analytic continuation to $\Cset\setminus\{1,3,5,7, \dots\}$ otherwise:
\begin{equation} \label{V.s}
V_s(\mathbb{S}) := \frac{2^{-s}\gammafcn((1-s)/2)}{\sqrt{\pi}\gammafcn(1-s/2)},  \quad s\in \Cset, s\neq1,3,5, \dots, \qquad V_0 := 1.
\end{equation}
Using properties of the digamma function and the Riemann zeta function, we obtain
\begin{align*}
\left.\frac{\dd V_s(\mathbb{S})}{\dd s} \right|_{s\to0} 
&= \left.\left\{ V_s(\mathbb{S}) \left[ \frac{1}{2} \left( \digammafcn(1 - s / 2) - \digammafcn((1-s)/2)  \right) - \log 2 \right] \right\} \right|_{s\to0} = 0, \\
\left. \frac{\dd }{\dd s} \left\{ \frac{2\zetafcn(s)}{(2\pi)^s} \, N^{1+s} \right\} \right|_{s\to0} 
&= \left. \left\{ \frac{2\zetafcn(s)}{(2\pi)^s} \, N^{1+s} \left[ \log N + \frac{\zetafcn^\prime(s)}{\zetafcn(s)} - \log ( 2 \pi ) \right] \right\} \right|_{s\to0} = - N \log N.
\end{align*}
Moreover, using that the Riemann zeta function vanishes at negative even integers and $B_{2k}^{(0)}(0) = 0$ for $k \geq 1$ (cf. \eqref{gen.BernoulliP.recc}), we obtain
\begin{equation*}
\left. \frac{\dd }{\dd s} \left\{ \frac{2 \alpha_n(s)}{(2\pi)^s} \, N^{1+s-2n} \, \zetafcn(s-2n)\right\} \right|_{s\to0} = 0 \qquad \text{for $n = 1, 2, 3, \dots$.}
\end{equation*}
Thus, putting everything together, we arrive at
\begin{equation*}
\mathcal{E}_{\mathrm{log}}( \mathbb{S}; N ) = \left. \frac{\dd \mathcal{L}_s(N)}{\dd s} \right|_{s \to 0} = - N \log N,
\end{equation*}
as we expected. 

A second fundamental observation is that the coefficients of the powers of $N$ in the asymptotic expansion of $\mathcal{L}_s(N)$ can be extended to meromorphic functions on $\mathbb{C}$. In this way one sees that the leading two terms ``swap places'' when moving $s$ from the interval $(-2,1)$ to interval $(1,3)$. (Note that we assume that $\mathcal{L}_s(N)=N(N-1)$ at $s=0$.) Thus, each term's singularity at $s=1$ can be avoided by moving away from the real line. In other words, the second term of the asymptotics of $\mathcal{L}_s(N)$ in the finite energy case $s \in (-2,1)$ is the analytic continuation of the leading term in the hyper-singular case $s\in (1,3)$ and vice versa.

This is the motivation for  trying the same ideas for the $d$-sphere in Section~\ref{section4}.

\section{Asymptotics of optimal Riesz $s$-energy}
\label{section3}
\subsection{The dominant term}
The leading term of the asymptotic expansion (as $N \to \infty$) of the maximal (if $s<0$) and minimal (if $s>0$) Riesz $s$-energy of  a compact set $A$ in $\mathbb{R}^{p}$ is well understood if $A$ has finite $s$-energy (that is, positive $s$-capacity). For sets of vanishing $s$-capacity the leading term is rather well understood in the sense that the existence of the coefficient has been established for a large class of compact sets by the second and third authors, but the determination of this coefficient for $s>d$ (except for one-dimensional sets) has to date remained a challenging open problem. 

The general theory for (the energy integral associated with) the {\em continuous $s$-potential} ($s<0$) covered in Bj{\"o}rck~\cite{Bj1956} provides Frostman-type results, existence and uniqueness results for the equilibrium measure $\mu_s$, and characterization of the support of $\mu_s$ for general compact sets $A$ in $\mathbb{R}^p$. Of particular interest is the observation that the support of $\mu_s$ is concentrated in the extreme points of the convex hull of $A$ if $s < -1$ and for $s < -2$ any maximal distribution (there is no unique equilibrium measure anymore for $s\leq-2$) consists of no more than $(p+1)$ point masses. 
% This limits the range of the potential theoretical regime to $-2<s<d$. 
The {\em singular Riesz $s$-potential on $\mathbb{S}^d$, $0<s<d$,} is the subject of classical potential theory (see, for example, Landkof~\cite{La1972}) with the value of the energy integral for $\mathbb{S}^d$ being already studied by P{\'o}lya and Szeg{\H{o}}  \cite{PoSz1931}. 
The range of the potential theoretical regime is thus limited to $-2<s<d$.  

The first results for the $d$-sphere in the {\em hypersingular case $s \geq d$} were given in \cite{KuSa1998}. This included the 
leading term for the exceptional case $s=d$, bounds for the leading term for $s>d$, separation results, and conjectures for the leading coefficient for the $2$-sphere. In subsequent work the existence of the leading term for $s> d$ was proven for the class of  $d$-rectifiable sets (\cite{HaSa2004,HaSa2005}) as well as for weighted Riesz $s$-energy on such sets (\cite{BoHaSa2008}).

\subsubsection*{The potential-theoretic regime $-2<s<d$}

A standard argument from classical potential theory involving the concept of {\em transfinite diameter} yields that the positive {\em $N$-th generalized diameter of $\mathbb{S}^d$}, $D_N^s( \mathbb{S}^d ) \DEF N ( N - 1 ) / \mathcal{E}_s( \mathbb{S}^d; N )$, forms a monotonically decreasing (increasing) sequence\footnote{This monotonicity holds more generally for any compact set $A$.}  bounded from below (above) if $s>0$ ($s<0$). This implies the existence of the limit $D^s \DEF \lim_{N \to \infty} D_N^s( \mathbb{S}^d )$ called the {\em generalized transfinite diameter of $\mathbb{S}^d$} introduced by P{\'o}lya and Szeg{\H{o}} in \cite{PoSz1931}, where it is further shown that the $s$-capacity and generalized transfinite diameter are equal.

\begin{thm} \label{prop:potential.theoret.case}
Let $d \geq 1$. Then for $-2 < s < d$ 
\begin{equation*}
\lim_{N \to \infty} \mathcal{E}_s( \mathbb{S}^d; N ) / N^2 = V_s( \mathbb{S}^d ) = 1 / \CAP_s( \mathbb{S}^d ) = 1 / D^s.
\end{equation*}
\end{thm}

Thus, in the potential-theoretic regime the  dominant term grows like $N^2$ as $N\to \infty$ and its coefficient in the asymptotic expansion of the optimal $s$-energy of $\mathbb{S}^d$  encodes the  $s$-capacity as well as the generalized transfinite diameter $D^s$ of $\mathbb{S}^d$.

For future reference, we remark that $V_s(\mathbb{S}^d)$ can be analytically extended to the complex $s$-plane except at the simple poles at $s = d +2k$, $k=0,1,2,\ldots$ if $d$ is odd and  for $k=0,\ldots, \frac{d}{2}-1$ if $d$ is even.  The residue  at $s=d+2k$, regardless of the parity of $d$, is
\begin{equation}
  (-1)^k 2^{-2k-1} \gammafcn( ( d + 1 ) / 2 ) \big/ \left[ \sqrt{\pi} \, k! \gammafcn( d/2 - k ) \right].
\end{equation}
We will  denote this meromorphic extension with the same symbol $V_s(\mathbb{S}^d)$.

\subsubsection*{The hypersingular case $s \geq d$} Since for $s \geq d$ the $s$-energy integral for every positive Borel probability measure supported on $\mathbb{S}^d$ is $+\infty$, potential theory fails to work. (The boundary or exceptional case $s=d$ can   still be treated using a particular normalization of the energy integral and a limit process as $s$ approaches $d$ from below, see \cite{CaHa2009}). %Methods from geometrical measure theory are used in order to prove the following results. 

The dominant term of the asymptotic expansion of the minimal $s$-energy grows like $N^2 \log N$ in the boundary case $s = d$. The coefficient is also known.
\begin{thm}[boundary case $s=d$, \cite{KuSa1998}]
\label{prop:1st.term.s.EQ.d}
\begin{equation} \label{eq:gammad.const}
\lim_{N \to \infty} \frac{\mathcal{E}_d( \mathbb{S}^d; N)}{N^2 \log N} = \frac{\mathcal{H}_d(\mathbb{B}^d)}{\mathcal{H}_d(\mathbb{S}^d)} =  \frac{1}{d} \frac{\gammafcn((d+1)/2)}{\sqrt{\pi} \gammafcn(d/2)}.
\end{equation}
\end{thm}
 Here and hereafter, $\mathcal{H}_d(\cdot)$ denotes  $d$-dimensional Hausdorff measure in $\mathbb{R}^p$, $p\geq d$, normalized such that a $d$-sided cube with side length $1$ has $\mathcal{H}_d$-measure equal to $1$.  In particular, $\mathcal{H}_d(\mathbb{B}^d)$ denotes the volume of the unit ball in ${\mathbb R}^d$ while $\mathcal{H}_d(\mathbb{S}^d)$ denotes the surface area of the unit sphere in ${\mathbb R}^{d+1}$.

In \cite{KuSa1998} the order of the growth rate of $\mathcal{E}_s( \mathbb{S}^d; N)$ was established: there are constants $C_1, C_2 > 0$ such that 
\begin{equation}\label{EsSd}
C_1 \, N^{1+s/d} \leq \mathcal{E}_s( \mathbb{S}^d; N ) \leq C_2 \, N^{1+s/d} \qquad (s > d \geq 2).
\end{equation}
The second and third author~\cite{HaSa2005} showed that the limit of the sequence $\mathcal{E}_s( \mathbb{S}^d; N ) / N^{1+s/d}$, indeed, exists.
More generally, the following result holds, which has been referred to as the {\em Poppy-seed Bagel Theorem} because of its interpretation for distributing points on a torus:  

\begin{thm}[\cite{HaSa2005,BoHaSa2008}] \label{prop:1st.term.s.GT.d}
Let $d \geq 1$ and $A\subset \Rset^p$ an infinite compact $d$-rectifiable set. Then for $s > d$
\begin{equation}
\label{eq:C.s.d}
\lim_{N \to \infty} \mathcal{E}_s( A; N ) \big/ N^{1+s/d} = C_{s,d} \big/ \left[ \mathcal{H}_d(A) \right]^{s/d},
\end{equation}
where $C_{s,d}$ is a finite positive constant (independent of  $A$).
\end{thm}
By {\em $d$-rectifiable set} we mean the Lipschitz image of a bounded set in $\Rset^d$.  

In particular, $$\lim_{N \to \infty} \frac{\mathcal{E}_s( {\mathbb S}^d; N ) }{ N^{1+s/d}} = \frac{C_{s,d}} {\mathcal{H}_d({\mathbb S}^d)^{s/d}}=C_{s,d}\left[\frac{\Gamma\left(\frac{d+1}{2}\right)}{2\pi^{(d+1)/2}}\right]^{s/d}, \qquad s>d.$$

 In \cite[Thm.~3.1]{MaMaRa2004} it is shown that $C_{s,1} = 2 \zetafcn(s)$.   For $d \geq 2$ the precise value of $C_{s,d}$ is not known. The significance (and difficulty of determining $C_{s,d}$) is deeply rooted in the connection to densest sphere packings. Let $\delta_N$ denote the best-packing distance of $N$-point configuration on $\mathbb{S}^d$. It is shown in \cite{BoHaSa2007} that
\begin{equation}\label{Cinfinity}
\lim_{s\to\infty} [ C_{s,d} ]^{1/s} = 1 / C_{\infty,d}, \qquad C_{\infty,d} \DEF \lim_{N \to \infty} N^{1/d} \, \delta_N = 2 \, [ \Delta_d / \mathcal{H}_d( \mathbb{B}^d ) ]^{1/d},
\end{equation}
where $\Delta_d$ is the largest sphere packing density in $\mathbb{R}^d$, which is only known for three cases: $\Delta_1=1$, $\Delta_2=\pi/\sqrt{12}$ (Thue in 1892 and L. Fejes T{\'o}th~\cite{Fe1972}), and $\Delta_3=\pi/\sqrt{18}$ (Kepler conjecture proved by Hales~\cite{Ha2005}). 
%Recently, Cohn and Kumar~\cite{CoKu2009} proved that the $E_8$ lattice   and the Leech lattice  are the unique densest {\em lattice} packings in, respectively, in $\mathbb{R}^8$ and $\mathbb{R}^{24}$  and showed that no sphere-packing in $\mathbb{R}^{8}$ ($\mathbb{R}^{24}$) can have density greater than $1 + 10^{-14}$ ($1 + 1.65 \times 10^{-30}$) times the density of the $E_8$ (Leech) lattice. For additional background we refer to \cite{CoEl2003}.

% of the limit $( C_{s,d} )^{1/s}$

%For $d=2$ more information is available. In \cite{KuSa1998} the following upper bound is given: 
%\begin{equation} \label{kuijlaars.saff.limit}
%\limsup_{N \to \infty} \frac{\mathcal{E}_s( \mathbb{S}^2; N )}{N^{1+s/2}} \leq \frac{\left( \sqrt{3} / 2 \right)^{s/2} \zetafcn_{\Lambda}(s)}{\left( 4 \pi \right)^{s/2}},
%\end{equation}
%where $\zetafcn_{\Lambda}(s)$ is the zeta function associated with the lattice $\Lambda$ consisting of points of the form $m(0,1) + n (1/2,\sqrt{3}/2)$ for $m$ and $n$ are integers. It is also conjectured in \cite{KuSa1998} that equality holds in \eqref{kuijlaars.saff.limit}. 
 
It is not difficult to see that the Epstein zeta function
$\zeta_\Lambda(s)$  for a lattice $\Lambda\subset \Rset^d$  defined for $s>d$ by
$$
\zeta_\Lambda(s):=\sum_{0\neq \PT{x}\in \Lambda}|\PT{x}|^{-s}
$$
yields an upper bound for $C_{s,d}$ as we now show. Let $\Omega$ denote a fundamental parallelotope for $\Lambda$. 
 For $n\in \Nset$,   the intersection
$X_n$ of $\Omega$  and the scaled lattice $(1/n)\Lambda$ contains
exactly $N=n^d$ points  and, for $s>d$, we have
$$
E_s(X_n)=\sum_{\PT{x}\in X_n}\sum_{\PT{x}\neq \PT{y}\in X_n}|\PT{x}-\PT{y}|^{-s}\le
\sum_{\PT{x}\in X_n}\sum_{\PT{x}\neq \PT{y}\in (1/n)\Lambda}|\PT{x}-\PT{y}|^{-s}=
n^{d+s}\zeta_\Lambda(s)=N^{1+s/d}\zeta_\Lambda(s).
$$
Referring  to \eqref{eq:C.s.d} with $A=\Omega$,  one obtains the following result.
\begin{prop} For $s>d$, 
\begin{equation}\label{Clatt}
C_{s,d}\le \min_{\Lambda}|\Lambda|^{s/d}\zeta_\Lambda(s),
\end{equation}
where the minimum is taken over all   lattices $\Lambda\subset\Rset^d $
with  covolume $|\Lambda|>0$. 
\end{prop}
In particular, as shown in \cite{KuSa1998}, 
$$
C_{s,2}\le \left(\frac{\sqrt{3}}{2}\right)^{s/2}\zeta_{\Lambda_2}(s) \ \text{ and } \lim_{N\to\infty}\frac{\mathcal{E}_s({\mathbb S}^2, N)}{N^{1+s/2}}\le \frac{(\sqrt{3}/2)^{s/2}\zeta_{\Lambda_2}(s)}{(4\pi)^{s/2}}, \qquad s>2, 
$$
where $\Lambda_2$ is the hexagonal lattice   consisting of points of the form $m(0,1) + n (1/2,\sqrt{3}/2)$ with $m$ and $n$  integers.

   For most values of $d$ we do not expect  equality to hold in \eqref{Clatt}; that is, we do not expect lattice packings to be optimal (especially for $d$ large where it is
expected that best packings are highly `disordered' and far from being lattice
arrangements, cf. \cite{ToSt2006}); however, recent results and conjectures of Cohn and Elkies \cite{CoEl2003},  Cohn and Kumar \cite{CoKu2007} and Cohn, Kumar, and Sch\"urmann \cite{CoKuSc2009} suggest that
 equality holds in \eqref{Clatt} for  $d=2, 4, 8$ and 24 leading to the following:
\begin{conj} \label{conj:C.s.d} For   $d=2, \, 4, \, 8$ and $24$,
$\displaystyle C_{s,d}=|\Lambda_d|^{s/d}\zeta_{\Lambda_d} (s)$ for $s>d$,
where  $\Lambda_d$ denotes, respectively,    the hexagonal lattice, $D_4$, $E_8$, and the Leech lattice.  
\end{conj}

 For $d=2$, this conjecture appears in  \cite{KuSa1998}.

\subsection{The second-order term} The only known results so far are estimates of the difference $\mathcal{E}_s( \mathbb{S}^d; N) - V_s(\mathbb{S}^d)N^2$ in the potential-theoretic regime. At the end of this section we present a lower bound for the optimal $d$-energy, the first hypersingular case. 

The {\em average distance problem} ($s=-1$) on a sphere was studied by Alexander~\cite{Al1972,Al1977} ($d=2$) and Stolarsky~\cite{St1973} ($d\geq2$, see citations therein for earlier work), later by Harman~\cite{Ha1982} ($d\geq2$, correct order and signs of the bounds up to $\log N$ factor in upper bound) and Beck~\cite{Be1984} ($d\geq2$, settled the correct order and signs of the bounds); generalized sums of distances ($-2< s<0$) were studied by Stolarsky~\cite{St1972} on $d$-spheres and by Alexander and Stolarsky~\cite{AlSt1974}, Alexander~\cite{Al1975} for general compact sets in $\mathbb{R}^p$ and general energy functionals; eventually, Wagner (\cite{Wa1990} upper bound and \cite{Wa1992} lower bound) arrived at  
\begin{equation*}
- \alpha  \, N^{1+s/d} \leq \mathcal{E}_s( \mathbb{S}^d; N ) - V_s( \mathbb{S}^d ) \, N^2 \leq - \beta \, N^{1+s/d} \qquad ( -2 < s < 0 ).
\end{equation*}
where  $\alpha$ and $\beta$ are positive constants depending on $s$ and $d$ but not $N$.

For the singular Riesz $s$-potential on $\mathbb{S}^2$, Wagner~\cite{Wa1992} found the upper bound
\begin{equation*}
\mathcal{E}_s( \mathbb{S}^2; N ) \leq V_s( \mathbb{S}^2 ) \, N^2 - C_2 \, N^{1+s/2}, \qquad \text{$C_2 > 0$ ($0 < s < 2$).}
\end{equation*}
The method of the alternative proof in \cite{RaSaZh1994} was generalized in \cite{KuSa1998} leading to
\begin{equation*}
\mathcal{E}_s( \mathbb{S}^d; N ) \leq V_s( \mathbb{S}^d ) \, N^2 - C_2^\prime \, N^{1+s/d}, \qquad C_2^\prime > 0 \qquad (0<s<d).
\end{equation*}
Wagner~\cite{Wa1990} also gave the lower bounds
\begin{equation*}
\mathcal{E}_s( \mathbb{S}^d; N ) \geq V_s( \mathbb{S}^d ) \, N^2 - \begin{cases} C_1 \, N^{1 + s / d} & d - 2 < s < d, \\   C_1^\prime \, N^{1 + s / (s+2)} & 0 < s \le  d - 2,\ d\ge 3,  \end{cases}
\end{equation*}
which were improved by the first author \cite{Br2006}. All results combined together lead to the correct order of   growth for the second-order term:
\begin{prop}
Let $d \geq 2$. Then for each $-2 < s < d$, there exist constants $c, C > 0$ which depend on $s$ and $d$, such that
\begin{equation*}
- c \, N^{1+s/d} \leq \mathcal{E}_s( \mathbb{S}^d; N ) - V_s( \mathbb{S}^d ) \, N^2 \leq - C \, N^{1+s/d}, \quad N\ge 2.
\end{equation*}
\end{prop}

Next, we present bounds for the hypersingular case $s\geq d$ which follow from a careful inspection of the proof of the dominant term in \cite{KuSa1998}. See also \cite{Br2006} for the potential-theoretic case $0<s<d$.

\begin{prop}
\label{thm:2nd.term.boundary}
Let $d \geq 2$. Then, as $N \to \infty$,
\begin{equation*} 
- c(d) \, N^2 + \mathcal{O}(N^{2-2/d} \log N)\le \mathcal{E}_d( \mathbb{S}^d; N )-\frac{\mathcal{H}_d(\mathbb{B}^d)}{\mathcal{H}_d(\mathbb{S}^d)} \, N^2 \log N  \leq   \frac{\mathcal{H}_d(\mathbb{B}^d)}{\mathcal{H}_d(\mathbb{S}^d)} \, N^2 \log \log N + \mathcal{O}(N^2),  
\end{equation*}
where the constant $c(d)$ is given by 
\begin{equation*}
c(d) \DEF \frac{\mathcal{H}_d(\mathbb{B}^d)}{\mathcal{H}_d(\mathbb{S}^d)} \left\{ 1 - \log \frac{\mathcal{H}_d(\mathbb{B}^d)}{\mathcal{H}_d(\mathbb{S}^d)} + d \left[ \digammafcn(d/2) - \digammafcn(1) - \log 2 \right] \right\} > 0.
\end{equation*}
(Recall that  $\digammafcn $ denotes the digamma function.)
\end{prop}

\begin{rmk}
For $d = 2$ one has $c(2) = 1 / 4$ and an $\mathcal{O}(N)$ term in the lower bound instead of $\mathcal{O}(N \log N)$.  \end{rmk}

The proof of Proposition~\ref{thm:2nd.term.boundary} is given in Section~\ref{sectionproofs} along with proofs of other new results stated in this section.

Following the approach leading to Proposition~\ref{thm:2nd.term.boundary}  we obtain for  $s>d$ the following crude estimate, which, curiously, reproduces the conjectured second term but only provides a lower bound for the leading term (that is, for the constant $C_{s,d}$ in the leading term).

\begin{prop} \label{thm:hypersing.lower.bound}
Let $d \geq 2$ and $s > d$. Then, for $(s-d)/2$ not an integer,
\begin{equation*}
\mathcal{E}_s( \mathbb{S}^d; N ) \geq A_{s,d} \, N^{1+s/d} + V_s(\mathbb{S}^d) \, N^2 + \mathcal{O}(N^{1+s/d-2/d}) \qquad \text{as $N \to \infty$,}
\end{equation*}
where 
\begin{equation*}
A_{s,d} = \frac{d}{s-d} \left[ \frac{1}{2} \frac{\gammafcn((d+1)/2) \gammafcn(1+(s-d)/2)}{\sqrt{\pi}\gammafcn(1+s/2)} \right]^{s/d}.
\end{equation*}
\end{prop}

Note that for $s > d + 2$ in the above proposition the $\mathcal{O}(N^{1+(s-2)/d})$ term dominates the $N^2$-term.

The upper bound of $\mathcal{E}_s(\mathbb{S}^d; N)$ in the hypersingular case $s>d$ is in the spirit of Proposition~\ref{thm:hypersing.lower.bound}.

\begin{prop} \label{thm:hypersing.upper.bound}
Let $d \geq 2$ and $s > d$. Then, for $(s-d)/2$ not an integer,
\begin{equation*}
\mathcal{E}_s( \mathbb{S}^d; N) \leq \left[ \frac{\mathcal{H}_d(\mathbb{B}^d)}{\mathcal{H}_d(\mathbb{S}^d)( 1-d/s )}   \right]^{s/d} \, N^{1+s/d} + \frac{s}{d} V_s(\mathbb{S}^d) \, N^2 + \mathcal{O}(  N^{1 + s / d - 2 / d} ).
\end{equation*}
\end{prop}

The constant $A_{s,d}$ in Proposition~\ref{thm:hypersing.lower.bound} and the corresponding constant in Proposition~\ref{thm:hypersing.upper.bound} provide lower and upper bounds for the  constant $C_{s,d}$ appearing in Proposition~\ref{prop:1st.term.s.EQ.d}.

\begin{cor}
Let $d \geq 2$ and $s > d$. Then for $(s-d)/2$ not an integer
\begin{equation*}
A_{s,d} \leq \frac{C_{s,d} }{ \left[ \mathcal{H}_d(\mathbb{S}^d) \right]^{s/d}} \leq \left[ \frac{\mathcal{H}_d(\mathbb{B}^d)}{\mathcal{H}_d(\mathbb{S}^d)( 1-d/s )}   \right]^{s/d},
\end{equation*}
where $A_{s,d}$ is given in Proposition~\ref{thm:hypersing.lower.bound}.
\end{cor}

For $d = 2$, the above bounds for $C_{s,d} / [ \mathcal{H}_d(\mathbb{S}^d) ]^{s/d}$ reduce to 
\begin{equation} \label{eq:bounds.4.cs2}
2^{-s/2} s^{-s/2} \big/ \left[ s/2 - 1 \right] \leq C_{s,2} / \left[4\pi \right]^{s/2} \leq 2^{-s} \left[ s / \left( s - 2 \right) \right]^{s/2}.
\end{equation}

\section{Conjectures for the Riesz $s$-energy}
\label{section4}
A   straightforward generalization of the asymptotics for the unit circle \eqref{circle.asympt} would be
\begin{equation}\label{straightforward}
\begin{split}
\mathcal{E}_s(\mathbb{S}^d; N) 
&= V_s( \mathbb{S}^d ) \, N^2 + \frac{C_{s,d}}{\left[\mathcal{H}_d(\mathbb{S}^d)\right]^{s/d}} \, N^{1+s/d} + \sum_{n=1}^p\beta_n(s,d) N^{1+s/d-2n/d} \\
&\phantom{=\pm}+ \mathcal{O}_{s,d,p}(N^{1+  s/d -2p/d-2/d}), \qquad \text{as $N \to \infty$}
\end{split}
\end{equation}
for $s$ not a pole of $V_s(\mathbb{S}^d)$ given in \eqref{V.s.d} and $C_{s,d}$ as defined in \eqref{eq:C.s.d}. The exceptional cases are caused by the simple poles of $V_s( \mathbb{S}^d)$ as a complex function in $s$. As $s \to d$, one of the terms in the finite sum has to compensate for the pole of the $N^2$ term, thus introducing a logarithmic term  (cf. the motivation for Conjecture~\ref{conj:Riesz.d.dsphere} in Section~\ref{sec:justification}). Recall  that for even $d$ there are only finitely many poles suggesting that there might only be finitely many exceptional cases for even $d$. %
It is possible that the higher-order terms ($p\geq1$) may not exist in this form;\footnote{In fact, Michael Kiessling from Rutgers University pointed out that higher-order terms might not exist, instead one could have an oscillating term as $N$ grows. Moreover, numerical results in Melnyk et al~\cite{MeKnSm1977} suggest that for $N$ fixed there might be an abrupt change of optimal configuration when an increasing $s$ passes certain critical values that depend on $N$. This behavior might also influence the existence of higher-order terms.} the presence of the $N^2$ and $N^{1+s/d}$ is motivated by known results and the principle of analytic continuation. In the potential theoretic case $-2<s<d$, the leading term is the $N^2$-term and $V_s(\mathbb{S}^d)$ equals the continuous $s$-energy of $\mathbb{S}^d$ (see Theorem~\ref{prop:potential.theoret.case}) and is its analytic continuation elsewhere. By the same token, the leading term in the hypersingular case $s>d$ is the $N^{1+s/d}$-term (see Theorem~\ref{prop:1st.term.s.EQ.d}) whose coefficient $C_{s,d}$ is assumed to have an analytic continuation to the complex $s$-plane as well.   % (see Conjecture~\ref{conj:C.s.d} tracking the supposed properties of $C_{s,d}$).

% Let $V_s(\mathbb{S}^d)$ denote the Riesz $s$-energy of $\mathbb{S}^d$ if $-2 < s < d$, $s \neq 0$, and its meromorphic continuation to $\Cset$ otherwise (cf., for example, \cite{}): \COMMENT{ref}
% \begin{equation} %\label{V.s.d}
% V_s(\mathbb{S}^d) = 2^{d-s-1} \frac{\gammafcn((d+1)/2) \gammafcn((d-s)/2)}{\sqrt{\pi}\gammafcn(d-s/2)},  \quad s\in \Cset, s\neq d,d+2,d+4,d+6, \dots.
% \end{equation}
% (Note that $V_0 = 1$.)

\begin{conj} \label{conj:Riesz.d.sphere}
Let $d\geq 2$. For $-2 <   s < d + 2$, $s\neq d$, there is a constant $C_{s,d}$ such that
\begin{equation}\label{eqn:Riesz.d.sphere}
\mathcal{E}_s( \mathbb{S}^d; N ) = V_s( \mathbb{S}^d ) \, N^2 + \frac{C_{s,d}}{\left[ \mathcal{H}_d( \mathbb{S}^d ) \right]^{s/d}} \, N^{1+s/d} + o(N^{1+   s/d}) \qquad \text{as $N \to \infty$,}
\end{equation}
where, for $s>d$, the constant   $C_{s,d}$ is the same as that in \eqref{eq:C.s.d}.
Furthermore, for   $d=2, \, 4, \, 8$ and $24$, 
$\displaystyle C_{s,d}=|\Lambda_d|^{s/d}\zeta_{\Lambda_d} (s)$,   
where $\Lambda_d$ is as in Conjecture~\ref{conj:C.s.d}.  For  $d=2$, $-2<s<4$, and $s\neq 2$, \eqref{eqn:Riesz.d.sphere} reduces to  
\begin{equation}\label{conj:Riesz.2.sphere}
\mathcal{E}_s( \mathbb{S}^2; N ) = \frac{2^{1-s}}{2-s} \, N^2 + \frac{\left( \sqrt{3} / 2 \right)^{s/2} \zetafcn_{\Lambda_2}(s)}{\left( 4 \pi \right)^{s/2}} \, N^{1+s/2} + o( N^{1+ s/2}) \qquad \text{as $N \to \infty$.}
\end{equation}
 
\end{conj}
\begin{rmk}
Note that $s=0$ in \eqref{eqn:Riesz.d.sphere} does not refer to the logarithmic case but rather to $\mathcal{E}_0( \mathbb{S}^d; N ) = N (N - 1)$, from which we deduce that   $C_{0,d}=-1$.  
\end{rmk}

For $s$ in the range $-2 < s < d + 2$ ($s \neq d$),   the $N^2$-term and $N^{1+s/d}$-term  are the two leading terms, interchanging their role as dominant term as $s$ passes by the boundary case $s = d$, where $V_s(\mathbb{S}^d)$ has a pole.  
 Its next pole for $s>d$ occurs at $s=d+2$.  For higher values of $s$ ($s > d + 2$) the $N^2$-term is conjectured to  be dominated by other powers of $N$.

% Moreover, the ideas leading to this conjecture were extended to the range $0<s<2$ in the same paper.
% In light of the now known asymptotics \eqref{circle.asympt} we believe the following to be true but are not yet able to prove it.
  
\begin{rmk}
Starting with the assumption that the Riesz $s$-energy of $N$ points is approximately given by $N$ times the potential $\Phi$ (or ``point energy'') created by all other $N-1$ points at a given point and using a semicontinuum approximation\footnote{This technique can be found in old papers addressing problems in solid state physics, cf. \cite[p.~188f]{GoAd1960}.} to approximate $\Phi$, Berezin~\cite{Be1986} arrived at the plausible asymptotics
\begin{equation*}
\begin{split}
\mathcal{E}_s( \mathbb{S}^2; N ) 
&\approx N^2 \frac{2^{1-s}}{2-s} \left[ 1 - \left( n / N \right)^{1-s/2} \right] + N \left( \frac{N \sqrt{3}}{8 \pi} \right)^{s/2} \\
&\phantom{=\pm\times}\times \left[ \frac{6}{1^s} + \frac{6}{( \sqrt{3} )^s} + \frac{6}{2^s} + \frac{12}{( \sqrt{7} )^s} + \frac{6}{3^s} + \frac{6}{( 2 \sqrt{3} )^s} + \frac{12}{( \sqrt{13} )^s} + \cdots \right].
\end{split}
\end{equation*}
The denominators [without the power] are the first $7$ distances in the hexagonal lattice $\Lambda_2$ and the numerators give the number of nearest neighbors with the corresponding distance. Thus, $\Phi$ is approximated by the flat hexagonal lattice using the nearest neighbors up to level $7$ and the remaining $N-n$ points are approximated by a continuum), where $n$ is one plus the number of at most $7$-th nearest neighbors. The $\cdots$ indicates that the local approximation can be extended to include more nearest neighbors. In fact, for $s>2$ the square-bracketed expression is the truncated zeta function for the hexagonal lattice. Figure \ref{fig-1} shows the distance distribution function for numerical approximation of a local optimal $900$-point $1$-energy configuration (cf. Womersley~\cite{WoSl2003a}). The first few rescaled lattice distances are superimposed over this graph. Note the remarkable coincidences with the peaks of the distance distribution function.
\end{rmk}

\begin{figure}[ht]
\begin{center}
\includegraphics[scale=1]{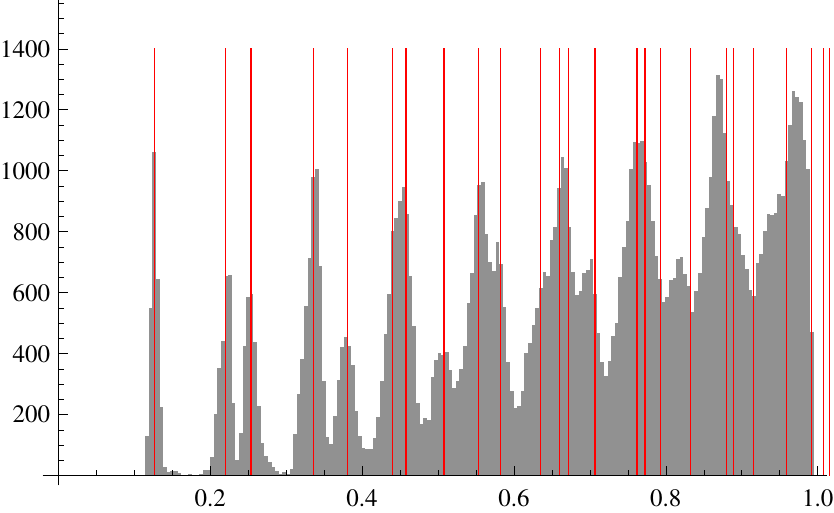}
\includegraphics[scale=1]{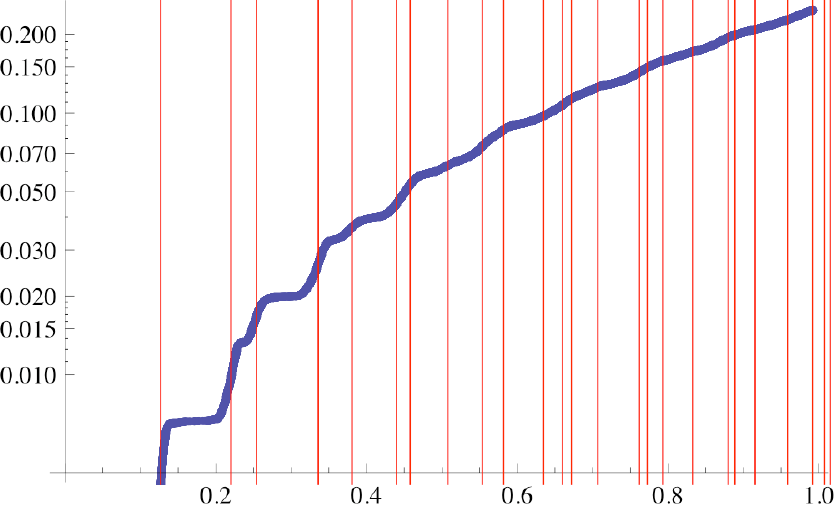}
\caption{\label{fig-1} The histogram (counting number of occurrences) and the empiric distribution function of the first $100 000$ distances of a local optimal $900$-point $1$-energy configuration. The vertical lines denote the hexagonal lattice distances adjusted such that the smallest distance coincidences with the best-packing distance.}
\end{center}
\end{figure}

 The zeta function $\zetafcn_{\Lambda_2}(s)$ appears in number theory as the zeta function of the imaginary quadratic field $\mathbb{Q}( \sqrt{-3} )$, whose integers can be identified with the hexagonal lattice $\Lambda_2$. It is known (cf., for example, \cite[Ch.~X, Sec.~7]{Co1980}) that $\zetafcn_{\Lambda_2}(s)$ admits a factorization
 \begin{equation} \label{zeta.lambda.prod}
 \zetafcn_{\Lambda_2}(s) = 6 \zetafcn( s / 2 ) \DirichletL_{-3}( s / 2 ), \qquad \re s > 2,
 \end{equation}
 into a product of the Riemann zeta function $\zetafcn$ and the first negative primitive Dirichlet $\DirichletL$-Series
 \begin{equation} 
 \DirichletL_{-3}(s) \DEF 1 - \frac{1}{2^s} + \frac{1}{4^s} - \frac{1}{5^s} + \frac{1}{7^s} - \cdots, \qquad \re s > 1.
 \end{equation}
 The Dirichlet $\DirichletL$-series above can be also expressed in terms of the Hurwitz zeta function $\zetafcn(s,a)$ at rational values $a$, cf. \cite{KuSa1998}. That is,
 \begin{equation} \label{dirichlet.id}
 \DirichletL_{-3}(s) = 3^{-s} \left[ \zetafcn(s,1/3) - \zetafcn(s,2/3) \right], \qquad \zetafcn(s,a) \DEF \sum_{k=0}^\infty \frac{1}{\left( k + a \right)^s}, \quad \re s > 1
 \end{equation}
 
\begin{rmk}
It is understood that $\zetafcn_{\Lambda_2}$ in \eqref{conj:Riesz.2.sphere} is the meromorphic extension to $\mathbb{C}$ of the right-hand side of \eqref{zeta.lambda.prod}. Since  $\zetafcn(s)$ is negative on the interval $[-1,1)$, has a pole at $s=1$, and is positive on $(1,\infty)$ and the Dirichlet $\DirichletL$-Series is positive on the interval $(-1,\infty)$\footnote{Note that $\DirichletL_{-3}(1-2m)=0$ for $m=1,2,3,4,\dots$, cf. \cite{MathWorldDirichletL2009}.} (cf.,  Eq. \eqref{zeta.lambda.prod}), it follows that $C_{s,2}$ in \eqref{conj:Riesz.2.sphere} would be negative for $-2<s<2$ and positive for $s>2$.
\end{rmk}

% \begin{equation*}
% \mathcal{E}_s( \mathbb{S}^2; N ) = \frac{2^{1-s}}{2-s} \, N^2 + \frac{\left( \sqrt{3} / 2 \right)^{s/2} \zetafcn_{\Lambda}(s)}{\left( 4 \pi \right)^{s/2}} \, N^{1+s/2} + \mathcal{O}_s( N^{-1+s/2}) \qquad \text{as $N \to \infty$.}
% \end{equation*}
% \begin{rmk} 
% 
% \end{rmk}
% 

% Figure \ref{fig0} shows the difference 
% \begin{equation*}
% \Delta_{1}(N) \DEF  \mathcal{E}_1( \mathbb{S}^2; N ) - 1 \times N^2 - \frac{\left( \sqrt{3} / 2 \right)^{1/2} \zetafcn_{\Lambda}(1)}{\left( 4 \pi \right)^{1/2}} \, N^{3/2},
% \end{equation*}
% where
% \begin{equation*}
% \frac{\left( \sqrt{3} / 2 \right)^{1/2} \zetafcn_{\Lambda}(1)}{\left( 4 \pi \right)^{1/2}} = -1.10610258671519\dots,
% \end{equation*}
% and the normalized difference $\Delta_1(N) / N$. The optimal energy is approximated using numerical data.
% \begin{figure}[ht]
% \begin{center}
% % \includegraphics[scale=.7]{conj1energy.eps} \includegraphics[scale=.7]{conj1energy2.eps} 
% \caption{\label{fig0}}
% \end{center}
% \end{figure}

Based on the motivating discussion in Section 7, we propose the following conjecture for the logarithmic energy.
\begin{conj} \label{conj:log.2.sphere}
For $d=2,4,8$, and $24$, 
\begin{equation}\label{logdasymp}
\mathcal{E}_{\mathrm{log}}( \mathbb{S}^d; N ) = V_{\mathrm{log}}( \mathbb{S}^d ) \, N^2 - \frac{1}{d} \, N \log N + C_{\mathrm{log},d} \, N + o(N), \qquad \text{as $N \to \infty$,}
\end{equation}
where the constant of the $N$-term is given by
\begin{equation}\label{Clogd}
C_{\mathrm{log},d} = \frac{1}{d}  \log (\mathcal{H}_d({\mathbb S}^d)/|\Lambda_d|)  + \zetafcn_{\Lambda_d}'(0).
\end{equation}
For the case $d=2$, \eqref{Clogd} reduces to
\begin{equation*}
C_{\mathrm{log},2} = 2 \log 2 + \frac{1}{2} \log \frac{2}{3} + 3 \log \frac{\sqrt{\pi}}{\gammafcn(1/3)} = -0.05560530494339251850\ldots < 0.
\end{equation*}
\end{conj}

\begin{rmk}
  We expect, more generally,  that \eqref{logdasymp} holds for arbitrary $d$ provided that $C_{s,d}$ is differentiable at $s=0$, in which case \eqref{Clogd} becomes $$C_{\log, d}=\left.\frac{\dd}{\dd s}C_{s,d}\right|_{s=0}+(1/d)\log \mathcal{H}_d({\mathbb S}^d).$$ 

 Regarding the bounds \eqref{liminfbound} and \eqref{limsupbound}, note that $C_{\mathrm{log},2}$ is closer to the upper bound given in  \eqref{limsupbound}, which gives rise to the question if the related argument can be improved to give the precise value.
 
\end{rmk}

\subsection{The boundary case $s=d$}

% From the result of Kuijlaars and the third author~\cite{KuSa1998} we obtain the dominant term of the asymptotical expansion of the optimal Riesz $d$-energy:
% \begin{equation*}
% \lim_{N \to \infty} \frac{\mathcal{E}_d( \mathbb{S}^d; N)}{N^2 \log N} = \frac{\mathcal{H}_d(\mathbb{B}^d)}{\mathcal{H}_d(\mathbb{S}^d)} = \frac{1}{d} \frac{\gammafcn((d+1)/2)}{\sqrt{\pi} \gammafcn(d/2)}.
% \end{equation*}
% (Here, $\mathbb{B}^d$ denotes the unit ball in $\mathbb{R}^d$.) 
As with the unit circle, we expect to obtain the asymptotics of the optimal Riesz energy in the singular case from the corresponding asymptotics of the Riesz $s$-energy, $s\neq d$ and $s$ sufficiently close to $d$, by means of a limit process $s \to d$. This approach leads to the next conjectures whose motivating analysis is given in Section 7. 

\begin{conj} \label{conj:Riesz.d.dsphere}
Let $d \ge 1$. Then
\begin{equation*}
\mathcal{E}_d( \mathbb{S}^d; N) = \frac{\mathcal{H}_d(\mathbb{B}^d)}{\mathcal{H}_d(\mathbb{S}^d)} \, N^2 \log N + C_{d,d} N^2 + \mathcal{O}(1) \qquad \text{as $N \to \infty$,}
\end{equation*}
where
\begin{equation*}
C_{d,d} = \lim_{s\to d} \left[ V_s( \mathbb{S}^d ) + \frac{C_{s,d}}{\left[ \mathcal{H}_d( \mathbb{S}^d ) \right]^{s/d}} \right].
\end{equation*}
For $d=2$, this becomes
\begin{equation*}
C_{2,2} = \frac{1}{4} \left[ \gamma - \log ( 2 \sqrt{3} \pi ) \right] + \frac{\sqrt{3}}{4 \pi} \left[ \gamma_1(2/3) - \gamma_1(1/3) \right] = -0.08576841030090248365\dots < 0.
\end{equation*}
Here, $\gamma$ is the Euler-Mascheroni constant and $\gamma_n(a)$ is the generalized Stieltjes constant appearing as the coefficient $\gamma_n(a) / n!$ of $(1-s)^n$ in the Laurent series expansion of the Hurwitz zeta function $\zetafcn(s,a)$   about $s=1$.
\end{conj}

%\begin{conj} \label{conj:Riesz.d.dsphere}
%For $d=2$ 
%\begin{equation*}
%\mathcal{E}_2( \mathbb{S}^2; N) = \frac{1}{4} \, N^2 \log N + C_{2,2} N^2 + \mathcal{O}(1) \qquad \text{as $N \to \infty$,}
%\end{equation*}
%where
%\begin{equation*}
%C_{2,2} = \frac{1}{4} \left[ \gamma - \log ( 2 \sqrt{3} \pi ) \right] + \frac{\sqrt{3}}{4 \pi} \left[ \gamma_1(2/3) - \gamma_1(1/3) \right] = -0.08576841030090248365\dots < 0.
%\end{equation*}
%Here, $\gamma$ is the Euler-Mascheroni constant and $\gamma_n(a)$ is the generalized Stieltjes constant appearing as the coefficient $\gamma_n(a) / n!$ of $(1-s)^n$ in the Laurent series expansion of $\zetafcn(s,a)$ about $s=1$.
%\end{conj}

% Figure \ref{fig2} shows the difference 
% \begin{equation*}
% \Delta_{2}(N) \DEF  \mathcal{E}_d( \mathbb{S}^d; N) - \frac{1}{4} \, N^2 \log N - C_{2,2} N^2
% \end{equation*}
% and the normalized difference $\Delta_{2}(N) / ( N \log N )$. The optimal energy is approximated using numerical data.
% \begin{figure}[ht]
% \begin{center}
% % \includegraphics[scale=.7]{conj2energy.eps} \includegraphics[scale=.7]{conj2energy2.eps} 
% \caption{\label{fig2}}
% \end{center}
% \end{figure}

\section{Numerical Results}

Rob Womersley from UNSW kindly provided numerical data, which we used to test our conjectures. For the logarithmic and the Coulomb cases ($s=1$) on ${\mathbb S}^2$ results for small numbers ($N = 4, \dots, 500$) and for large  numbers of points  ($N=(n+1)^2$ points, $N$ up to $22801$) are given. The reader is cautioned that these numerical data represent approximate optimal energies which we denote by $\hat{\mathcal{E}}_{\rm log}({\mathbb S}^2;N)$ or $\hat{\mathcal{E}}_s({\mathbb S}^2;N)$. 

A general observation is the slow convergence of the sequence of $s$-energy values. 

\subsection{Logarithmic case}

Figures~\ref{fig:2} and \ref{fig:3} show the convergence to the conjectured coefficient of the $N$-term (see Conjecture~\ref{conj:log.2.sphere})
\begin{equation}\label{log.energy.diff1}
\left\{ \hat{\mathcal{E}}_{\rm log}({\mathbb S}^2;N) - \left[ V_{\mathrm{log}}( \mathbb{S}^2 ) \, N^2 - \frac{1}{2} \, N \log N \right] \right\} / N.
\end{equation}
%The inset graph captures the behavior if the $N$-term is included, that is 
%\begin{equation}\label{log.energy.diff2}
%\left\{ \hat{\mathcal{E}}_{\rm log}({\mathbb S}^2;N) - \left[ V_{\mathrm{log}}( \mathbb{S}^2 ) \, N^2 - \frac{1}{2} \, N \log N + C_{\mathrm{log},2} \, N \right] \right\} / \log N.
%\end{equation}
The horizontal line indicates the value of $C_{\mathrm{log},2}$ given in Conjecture~\ref{conj:log.2.sphere}.

\begin{figure}[ht]
\begin{center}
\includegraphics[scale=1.2]{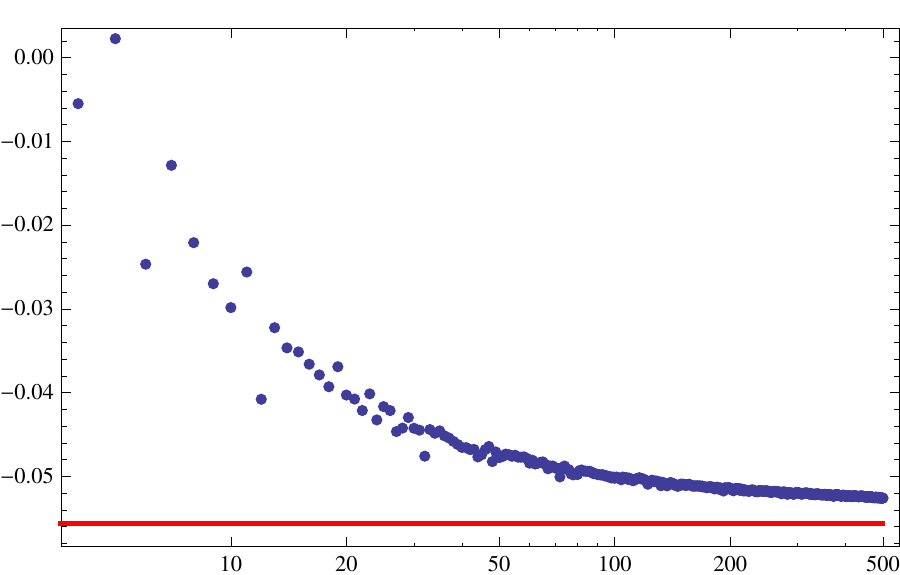}
\caption{\label{fig:2} The logarithmic case on $\mathbb{S}^2$, $N = 4, 5, 6 \dots, 500$ points. The horizontal axis shows $N$ on a logarithmic scale and the quantity \eqref{log.energy.diff1} on the vertical axis. The horizontal line shows the conjectured limit $C_{\log, 2}$.}
\end{center}
\end{figure}

\begin{figure}[ht]
\begin{center}
\includegraphics[scale=1.2]{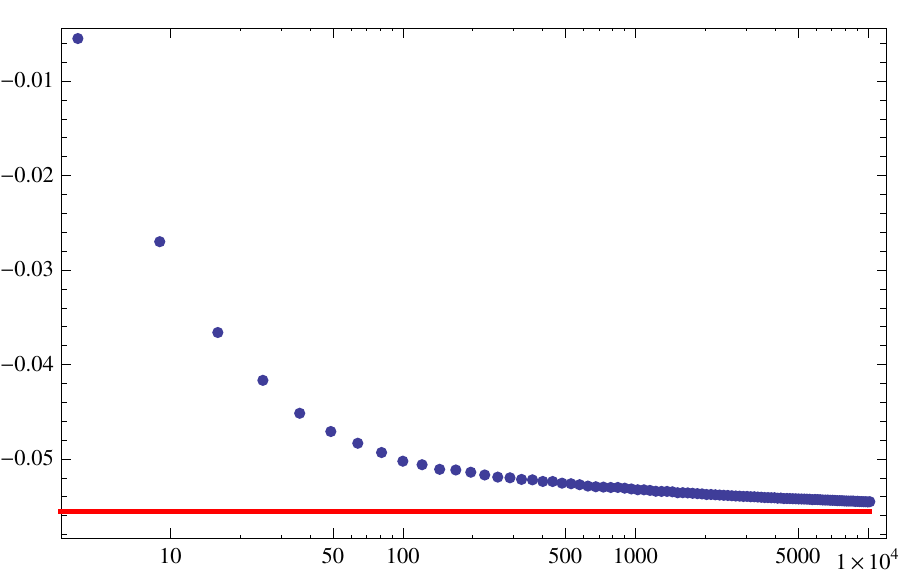}
\caption{\label{fig:3} The logarithmic case on $\mathbb{S}^2$ with $N = 4, 9, 16 \dots, 10201$ points. The same quantities are shown as in Figure~\ref{fig:2}.}
\end{center}
\end{figure}

\subsection{The Coulomb case $s=1$} 

Figures~\ref{fig:4} and \ref{fig:5} show the convergence to the conjectured coefficient of the $N^{1+1/2}$-term (see Conjecture~\ref{conj:Riesz.d.dsphere})
\begin{equation}\label{s1.diff1}
\left\{ \hat{\mathcal{E}}_1({\mathbb S}^2;N) - 1 \times N^2 \right\} / N^{1+1/2}.
\end{equation}
%The inset graph captures the behavior if this term is included, that is 
%\begin{equation}\label{s1.diff2}
%\left\{ \hat{\mathcal{E}}_1({\mathbb S}^2;N) - \left[ 1 \times N^2 + \frac{C_{1,2}}{\sqrt{4\pi}} N^{1+1/2} \right] \right\} / N^{1/2}.
%\end{equation}
The horizontal line indicates  the value of $C_{1,2} / \sqrt{4\pi}$.

\begin{figure}[ht]
\begin{center}
\includegraphics[scale=1.2]{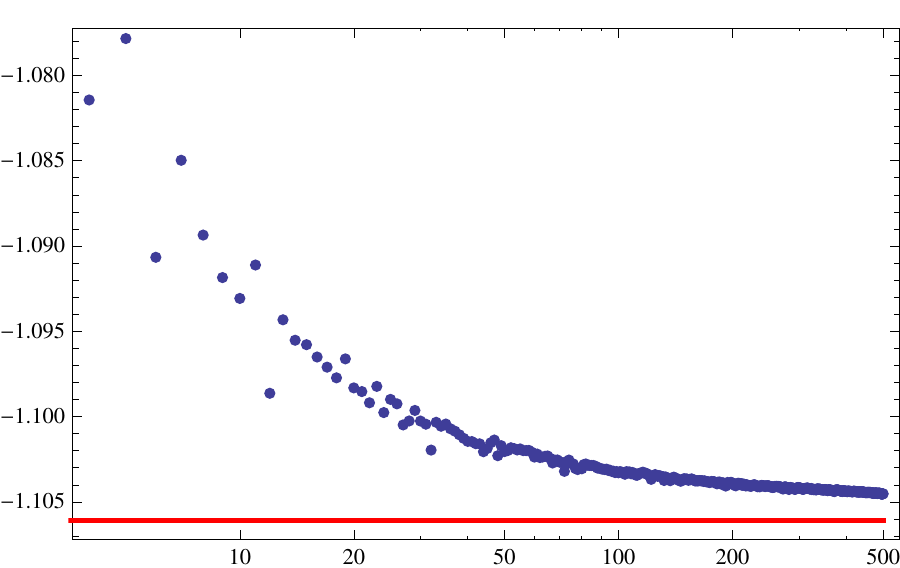}
\caption{\label{fig:4} The Coulomb case ($s=1$) on $\mathbb{S}^2$, $N = 4, 5, 6 \dots, 500$ points.   The horizontal axis shows $N$ on a logarithmic scale and the quantity \eqref{s1.diff1} on the vertical axis. The horizontal line shows the conjectured limit $C_{1,2} / \sqrt{4\pi}$.}
\end{center}
\end{figure}

\begin{figure}[ht]
\begin{center}
\includegraphics[scale=1.2]{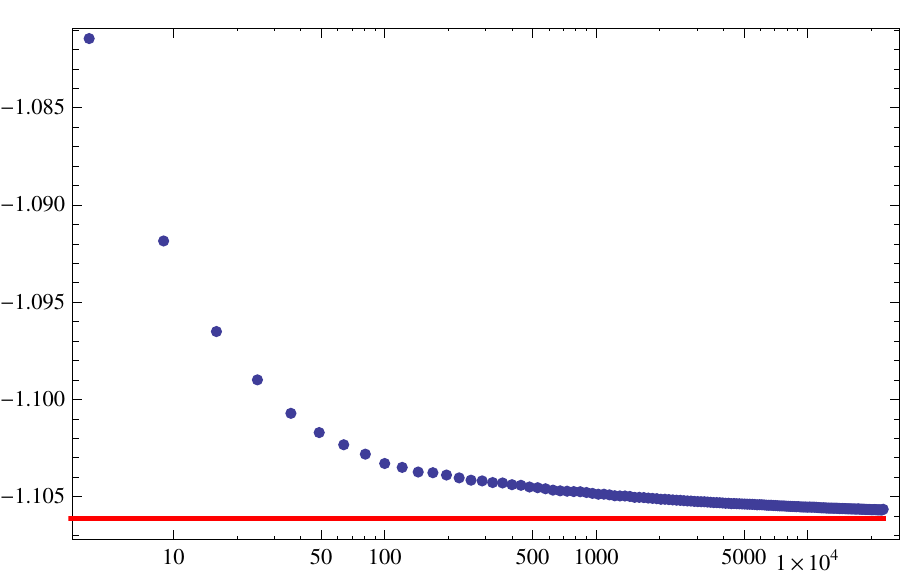}
\caption{\label{fig:5} The Coulomb case ($s=1$) on $\mathbb{S}^2$, $N = 4, 9, 16 \dots, 22801$ points. The same quantities are shown as in Figure~\ref{fig:4}.}
\end{center}
\end{figure}

\subsection{The boundary case $s = d = 2$}

Figure~\ref{fig:6} shows the convergence to the conjectured coefficient of the $N^{2}$-term (see Conjecture~\ref{conj:Riesz.d.dsphere})
\begin{equation}\label{s2.diff1}
\left\{ \hat{\mathcal{E}}_2({\mathbb S}^2;N) - \frac{1}{4} \, N^2 \log N \right\} / N^2.
\end{equation}
%The inset graph captures the behavior if this term is included, that is 
%\begin{equation}\label{s2.diff2}
%\left\{ \hat{\mathcal{E}}_2({\mathbb S}^2;N) - \left[ \frac{1}{4} \, N^2 \log N + C_{2,2} \, N^2 \right] \right\} / \left[ N \log N \right].
%\end{equation}
The horizontal line indicates the value of $C_{2,2}$.

\begin{figure}[ht]
\begin{center}
\includegraphics[scale=1.2]{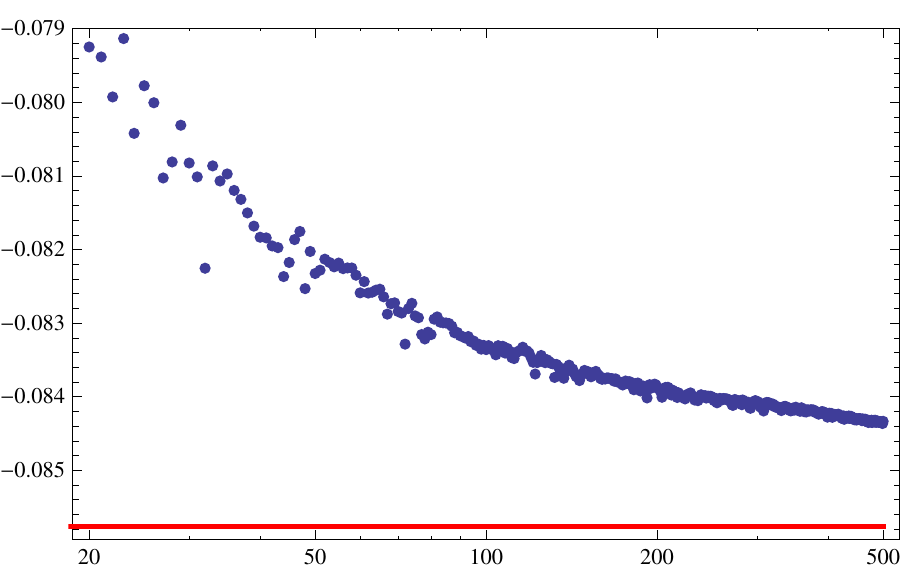}
\caption{\label{fig:6} The boundary case ($s=2$) on $\mathbb{S}^2$, $N = 4, 5, 6 \dots, 500$ points. The horizontal axis shows $N$ on a linear scale and the quantity \eqref{s2.diff1} on the vertical axis. The horizontal line shows the conjectured limit $C_{2,2}$.}
\end{center}
\end{figure}

\section{Proofs}
In the following we set $$\omega_d:=\mathcal{H}_d(\mathbb{S}^d)= \frac{2\pi^{(d+1)/2}}{\Gamma((d+1)/2)}.$$
\label{sectionproofs}
\begin{proof}[Proof of lower bound in Proposition~\ref{thm:2nd.term.boundary}]
This proof follows the first part of the proof of Theorem~3 in \cite{KuSa1998}. By an idea of Wagner, the (hyper)singular Riesz $d$-kernel $1/r^d$ is approximated by the smaller continuous kernel $1 / ( \eps + r^2 )^{d/2}$ ($\eps > 0$). Then, for $\PT{x}, \PT{y}\in \mathbb{S}^d$, we have $1 / ( \eps + |\PT{x}-\PT{y}|^2 )^{d/2}=K_\eps( \langle \PT{x},\PT{y}\rangle)$ where $K_\eps(t) \DEF ( 2 - 2 t + \eps )^{-d/2}$. Note that $K_\eps$ is positive definite in the sense of Schoenberg~\cite{Sch1942}; that is, it has an expansion $K_\eps(t) = \sum_{n = 0}^\infty a_n(\eps) \, P_n^d(t)$ in terms of ultraspherical (normalized Gegenbauer polynomials $P_n^d(t) = C_n^\lambda(t) / C_n^\lambda(1)$, where $\lambda = ( d - 1 ) / 2$) with positive coefficients $a_n(\eps)$ ($n \geq 1$) giving rise to the estimates
\begin{equation} \label{Edlowerbnd}
\begin{split}
E_d( X_N )  &\geq \sum_{j \neq k} \frac{1}{\left[ \eps + \left| \PT{x}_j - \PT{x}_k \right|^2 \right]^{d/2}} = \sum_{ j, k}  K_\eps( \langle \PT{x}_j, \PT{x}_k \rangle )-\sum_{j}  K_\eps( \langle \PT{x}_j, \PT{x}_j \rangle ) \\ &\geq a_0(\eps) \, N^2 - K_\eps(1) \, N,
\end{split}
\end{equation}
where we used the positivity of $K_\eps$ and $P_0^d(t)=1$ (cf. Sec. 3 of \cite{KuSa1998}).  Using the integral representation \cite[Eq.~15.6.1]{DLMF2010.05.10} of the regularized Gauss hypergeometric function we obtain  
 \begin{align*}
a_0( \eps ) 
&= \frac{\omega_{d-1}}{\omega_d} \int_{-1}^1 \left( 2 - 2 t + \eps \right)^{-d/2} \left( 1 - t^2 \right)^{d/2 - 1} \dd t \notag \\
&= \left( 4 + \eps \right)^{-d/2} 2^{d-1} \frac{\omega_{d-1}}{\omega_d} \int_0^1 u^{d/2-1} \left( 1 - u \right)^{d/2-1} \left( 1 - \frac{1}{1+\eps/4} u \right)^{-d/2} \dd u \notag \\
&= \left( 4 + \eps \right)^{-d/2} 2^{d-1} \frac{\omega_{d-1}}{\omega_d} \gammafcn(d/2)^2\HypergeomReg{2}{1}{d/2,d/2}{d}{\frac{1}{1+\eps/4}}.
\end{align*}
 Assuming that $\eps < 2$, we apply the linear transformation \cite[Eq.~15.8.11]{DLMF2010.05.10} with the understanding that $\digammafcn(d/2 - k) / \gammafcn(d/2 - k)$ is interpreted as $(-1)^{k-d/2+1} (k - d/2)!$ if $d/2-k$ is a non-positive integer: 
\begin{align*}
a_0( \eps ) 
% &= - \left( 4 + \eps \right)^{-d/2} 2^{d-1} \frac{\omega_{d-1}}{\omega_d} \gammafcn(d/2) \gammafcn(d/2) \, \frac{\left( 4 + \eps \right)^{d/2}}{\gammafcn(d/2)} \sum_{k=0}^\infty \frac{\Pochhsymb{d/2}{k}}{k! k! \gammafcn(d/2-k)} \left( \frac{\eps}{4} \right)^{k} \left[ \log \frac{\eps}{4} - 2 \digammafcn(k+1) + \digammafcn(d/2 + k) + \digammafcn(d/2 - k) \right]
&= - \frac{1}{2} \frac{\omega_{d-1}}{\omega_d} \gammafcn(d/2) \sum_{k=0}^\infty \frac{\Pochhsymb{d/2}{k}}{k! k! \gammafcn(d/2-k)} \left( \frac{\eps}{4} \right)^{k} \left[ \log \frac{\eps}{4} - 2 \digammafcn(k+1) + \digammafcn(d/2 + k) + \digammafcn(d/2 - k) \right] \\
&= \frac{1}{2} \frac{\omega_{d-1}}{\omega_d} \left( - \log \frac{\eps}{4} \right) \Hypergeom{2}{1}{1-d/2,d/2}{1}{- \frac{\eps}{4}} \\
&\phantom{=}- \frac{1}{2} \frac{\omega_{d-1}}{\omega_d} \gammafcn(d/2) \sum_{k=0}^\infty \frac{\Pochhsymb{d/2}{k}}{k! k! \gammafcn(d/2-k)} \left( \frac{\eps}{4} \right)^{k} \left[ \digammafcn(d/2 + k) + \digammafcn(d/2 - k) - 2 \digammafcn(k+1) \right].
\end{align*}
Note that the (non-regularized) Gauss hypergeometric function is a polynomial of degree $d/2-1$ if $d$ is even and reduces to $1$ if $d = 2$. Using the series representation \cite[Eq.~15.2.1]{DLMF2010.05.10} we arrive at 
\begin{align*}
a_0( \eps ) &= \frac{1}{2} \frac{\omega_{1}}{\omega_2} \left( - \log \eps \right) + \frac{\omega_{1}}{\omega_2} \log 2 + \mathcal{O}(\eps) = \frac{1}{4}\log(4/\eps) + \mathcal{O}(\eps) \quad \text{as $\eps \to 0$ if $d = 2$,} \\
a_0( \eps ) &= \frac{1}{2} \frac{\omega_{d-1}}{\omega_d} \left( - \log \eps \right) - \frac{\omega_{d-1}}{\omega_d} \left[ \digammafcn(d/2) - \digammafcn(1) - \log 2 \right] + \mathcal{O}(\eps \log(1/\eps)) \quad \text{as $\eps \to 0$ if $d \geq 3$.}
\end{align*}
Using the substitution $\eps = a^2 N^{-2/d}$ and noting that $K_\eps(1)=\eps^{-d/2}$, there follows 
from \eqref{Edlowerbnd} (one has $\mathcal{O}(N^{2-2/d})$ as $N \to \infty$ if $d = 2$)
\begin{equation*}
\mathcal{E}_d( \mathbb{S}^d; N ) \geq \begin{cases} \frac{1}{4}   \, N^2 \log N - f(2; a) \, N^2 + \mathcal{O}(N^{}), & d=2,\\
\frac{1}{d} \frac{\omega_{d-1}}{\omega_d} \, N^2 \log N - f(d; a) \, N^2 + \mathcal{O}(N^{2-2/d} \log N) , & d\ge 3,\end{cases}
\end{equation*}
as $N \to \infty$, where for each positive integer $d$ the function 
\begin{equation*}
f(d; a) \DEF \frac{\omega_{d-1}}{\omega_d} \left[ \log a + \digammafcn(d/2) - \digammafcn(1) - \log 2 \right] + a^{-d}, \qquad a > 0,
\end{equation*}
has a single global minimum at $a^* = [ \omega_{d-1} / (d\omega_d )]^{-1/d} $ in the interval $(0,\infty)$ with value
\begin{equation*}
f(d; a^*) = \frac{1}{d} \frac{\omega_{d-1}}{\omega_d} \left[ 1 - \log( \frac{1}{d} \frac{\omega_{d-1}}{\omega_d} ) + d \left( \digammafcn(d/2) - \digammafcn(1) - \log 2 \right) \right].
\end{equation*}
Let $F_d \DEF  \omega_{d-1} / (d\ \omega_d)$.  Then it is elementary to verify that   $F_d=\mathcal{H}_d(\mathbb{B}^d) / \mathcal{H}_d(\mathbb{S}^d)$. It remains to show that $f(d;a^*) > 0$. It is easy to see that $F_d > F_{d+2} > 0$ and $F_1 = 1 / \pi < 1$ and $F_2 = 1 / 4 < 1$. Thus, $1 - \log F_d > 0$. Since the digamma function is strictly  increasing, the expression $\digammafcn(d/2) - \digammafcn(1) - \log 2 > 0$ for $d \geq 4$ making $f(d;a^*) > 0$ for $d \geq 4$. Direct computations show that 
\begin{equation*}
f(1; a^*) = \left[ 1 + \log( \pi / 8 ) \right] / \pi, \qquad f(2; a^*) = 1 / 4, \qquad f(3; a^*) = 2 \left[ 7 + \log( 3 \pi / 1024 ) \right] / \left( 3 \pi \right)
\end{equation*}
are all positive, which completes the proof of the lower bound.
\end{proof}

We need the following auxiliary results for the upper bound in   Proposition~\ref{thm:2nd.term.boundary} . Let $C(\PT{x}, \rho)$ denote the spherical cap $\{ \PT{y} \in \mathbb{S}^d : | \PT{y} - \PT{x} | \leq \rho \} = \{ \PT{y} \in \mathbb{S}^d : \langle \PT{y} , \PT{x} \rangle \geq 1 - \rho^2 / 2 \}$.

\begin{lem} \label{lem:aux.1}
Let $d \geq 1$. Then for $\PT{x} \in \mathbb{S}^d$ and $0 < \rho \leq 2$, the normalized surface area measure of the spherical cap $C(\PT{x}, \rho)$ is given by
\begin{equation*}
\sigma_d( C(\PT{x}, \rho) ) = \frac{1}{d} \frac{\omega_{d-1}}{\omega_d} \, \rho^d \, \Hypergeom{2}{1}{1-d/2,d/2}{1+d/2}{\rho^2 / 4},
\end{equation*}
where the Gauss hypergeometric function is a polynomial if $d$ is even and reduces to $1$ if $d=2$.
\end{lem}

\begin{proof}
Using the definition of the spherical cap, the Funk-Hecke formula (see \cite{Mu1966}), and the substitution $t=1-(\rho^2/2)u$, the surface area of a spherical cap can be written in terms of a hypergeometric function as follows
\begin{align*}
\sigma_d( C(\PT{x}, \rho) ) 
&= \int_{C(\PT{x}, \rho)} \dd \sigma_d =  \frac{\omega_{d-1}}{\omega_d} \int_{1-\rho^2/2}^1 \left( 1 - t^2 \right)^{d/2-1} \dd t \\
&= \frac{1}{2} \frac{\omega_{d-1}}{\omega_d} \, \rho^d \, \int_0^1 u^{d/2-1} \left( 1 - u \right)^{1-1} \left[ 1 - \left(\rho^2 / 4\right) u \right]^{d/2-1} \dd u \\
&= \frac{1}{2} \frac{\omega_{d-1}}{\omega_d} \, \rho^d \, \frac{\gammafcn(d/2) \gammafcn(1)}{\gammafcn(1+d/2)} \Hypergeom{2}{1}{1-d/2,d/2}{1+d/2}{\rho^2 / 4}
\end{align*}
from which the result follows using properties of the Gamma function. 
\end{proof}

\begin{lem} 
\label{lem:identity}
Let $m$ be a positive integer. For $z \in \mathbb{C} \setminus \{ -1, -2, -3, \dots \}$ there holds
\begin{equation*} 
\sum_{k=0}^m \frac{\Pochhsymb{z}{k}\Pochhsymb{-z}{m-k}}{k!(m-k)!} \left[ \digammafcn(k + z) - \digammafcn(k+1) \right] = \frac{\Pochhsymb{1-z}{m}}{m! \, m}.
\end{equation*}
\end{lem}

\begin{proof}
Let $f_m(x)$ denote the sum for real $x>0$. Using the integral representation \cite[Eq.~2.2.4(20)]{PrBrMa1986I}
\begin{equation*}
\int_0^1 \frac{t^\alpha - t^\beta}{1-t} \dd t = \digammafcn(\beta+1) - \digammafcn(\alpha+1), \qquad \re \alpha, \re \beta > -1,
\end{equation*}
we obtain
\begin{equation*}
f_m(x) = \int_0^1 \frac{1-t^{x-1}}{1-t} g_m(x,t) \dd t, \qquad g_m(x,t) \DEF \sum_{k=0}^m \frac{\Pochhsymb{x}{k}\Pochhsymb{-x}{m-k}}{k!(m-k)!} t^k,
\end{equation*}
where the function $g_m(x,t)$ can be expressed as regularized Gauss hypergeometric functions
\begin{align*}
g_m(x,t) = (-1)^m \frac{\gammafcn(1+x)}{m!} \HypergeomReg{2}{1}{-m,x}{1+x-m}{t} = (-1)^m \frac{\gammafcn(1+x)}{m!} \left( 1 - t \right) \HypergeomReg{2}{1}{1-m,1+x}{1+x-m}{t}.
\end{align*}
Substituting the series expansion of $g_m(x,t)$ we have (for $x-m$ not a negative integer)
\begin{align*}
f_m(x) 
&= (-1)^m \frac{\gammafcn(1+x)}{m!} \sum_{k=0}^{m-1} \frac{\Pochhsymb{1-m}{k}\Pochhsymb{1+x}{k}}{\gammafcn(1+x-m+k) k!} \int_0^1 \left( t^k - t^{k+x-1} \right) \dd t \\
&= (-1)^m \frac{\gammafcn(1+x)}{m!} \sum_{k=0}^{m-1} \frac{\Pochhsymb{1-m}{k}\Pochhsymb{1+x}{k}}{\gammafcn(1+x-m+k) k!} \frac{1}{k+1} - (-1)^m \frac{\gammafcn(1+x)}{m!} \sum_{k=0}^{m-1} \frac{\Pochhsymb{1-m}{k}\Pochhsymb{1+x}{k}}{\gammafcn(1+x-m+k) k!} \frac{1}{k+x} \\
% &= (-1)^m \frac{\gammafcn(x)}{m! (- m)}  \sum_{k=1}^{m} \frac{\Pochhsymb{-m}{k}\Pochhsymb{x}{k}}{\gammafcn(x-m+k) k!} - (-1)^m \frac{\gammafcn(x)}{m!} \sum_{k=0}^{m-1} \frac{\Pochhsymb{1-m}{k}\Pochhsymb{x}{k}}{\gammafcn(1+x-m+k) k!}
&= (-1)^m \frac{1}{m! (- m)} \sum_{k=1}^{m} \frac{\Pochhsymb{-m}{k} \gammafcn(k+x)}{\gammafcn(x-m+k) k!} - (-1)^m \frac{1}{m!} \sum_{k=0}^{m-1} \frac{\Pochhsymb{1-m}{k}\gammafcn(k+x)}{\gammafcn(1+x-m+k) k!} \\
&= (-1)^m \frac{1}{m! (- m)} \frac{\gammafcn(x)}{\gammafcn(x-m)} \sum_{k=1}^{m} \frac{\Pochhsymb{-m}{k} \Pochhsymb{x}{k}}{\Pochhsymb{x-m}{k} k!} - (-1)^m \frac{1}{m!} \frac{\gammafcn(x)}{\gammafcn(x+1-m)} \sum_{k=0}^{m-1} \frac{\Pochhsymb{1-m}{k}\Pochhsymb{x}{k}}{\Pochhsymb{x+1-m}{k} k!} \\
&= (-1)^m \frac{1}{m! (- m)} \frac{\gammafcn(x)}{\gammafcn(x-m)} \left( \Hypergeom{2}{1}{-m,x}{x-m}{1} - 1 \right) - (-1)^m \frac{1}{m!} \frac{\gammafcn(x)}{\gammafcn(x+1-m)} \Hypergeom{2}{1}{1-m,x}{x+1-m}{1}. 
\end{align*}
Using the Chu-Vandermonde Identity \cite[Eq.~15.4.24]{DLMF2010.05.10} $\Hypergeom{2}{1}{-n,b}{c}{1} = \Pochhsymb{c-b}{n} / \Pochhsymb{c}{n}$, we obtain
\begin{align*}
f_m(x) 
&= (-1)^m \frac{1}{m! (- m)} \frac{\gammafcn(x)}{\gammafcn(x-m)} \left[ \frac{\Pochhsymb{-m}{m}}{\Pochhsymb{x-m}{m}} - 1 \right] - (-1)^m \frac{1}{m!} \frac{\gammafcn(x)}{\gammafcn(x+1-m)} \frac{\Pochhsymb{1-m}{m-1}}{\Pochhsymb{x+1-m}{m-1}} \\
&= (-1)^m \frac{1}{m! (- m)} \frac{\gammafcn(x)}{\gammafcn(x-m)} \left[ (-1)^m m! \frac{\gammafcn(x-m)}{\gammafcn(x)} - 1 \right] - \frac{1}{(-m)} \\
&= (-1)^m \frac{1}{m! \, m} \frac{\gammafcn(x)}{\gammafcn(x-m)} = \frac{1}{m! \, m} \frac{(-1)^m }{\Pochhsymb{x}{-m}} =  \frac{\Pochhsymb{1-x}{m}}{m!} \frac{1}{m}, % = \frac{\Pochhsymb{-x}{m}}{m!} \frac{m-x}{m(-x)} = \frac{\Pochhsymb{-x}{m}}{m!} \left( \frac{1}{m} - \frac{1}{x} \right),
\end{align*}
where in the last line we used properties of the Pochhammer symbol (see, for example, \cite[Appendix II.2]{PrBrMa1990III}. Since $f_m(x)$ is, in fact, analytic in $\mathbb{C}$ with poles at negative integers due to the digamma function (the singularity at $0$ can be removed), the identity
\begin{equation*}
f_m(x) = \frac{\Pochhsymb{1-x}{m}}{m! \, m}% = \frac{\Pochhsymb{-x}{m}}{m!} \left( \frac{1}{m} - \frac{1}{x} \right)
\end{equation*}
can be extended to $\mathbb{C} \setminus \{-1,-2,-3,\dots\}$ by analytic continuation.
\end{proof}

\begin{lem} \label{lem:aux.2}
Let $d \geq 1$. For $\PT{x} \in \mathbb{S}^d$ and $0< \rho < 2$
\begin{equation*}
\begin{split}
\int_{\mathbb{S}^d \setminus C(\PT{x},\rho)} \left| \PT{x} - \PT{y} \right|^{-d} \dd \sigma_d(\PT{y}) 
&= \frac{\omega_{d-1}}{\omega_d} \left( - \log \rho \right) - \frac{1}{2} \frac{\omega_{d-1}}{\omega_d} \left[ \digammafcn(d/2) - \digammafcn(1) - 2 \log 2 \right] \\
&\phantom{=}- \frac{1}{2} \frac{\omega_{d-1}}{\omega_d} \sum_{m=1}^\infty \frac{\Pochhsymb{1-d/2}{m}}{m! \, m} \left( \frac{\rho}{2} \right)^{2 m},
\end{split}
\end{equation*}
where the series terminates after finitely many terms if $d$ is even.
\end{lem}

\begin{proof}
Writing $| \PT{x} - \PT{y} |^{-d} = [2 ( 1 - t ) ]^{-d/2}$ with $t=\langle \PT{x},\PT{y}\rangle$ and using the substitution $1 + t = 2 ( 1 - \rho^2 / 4 ) u$ the integral can be expressed as a hypergeometric function as in the proof of Lemma~\ref{lem:aux.1}:
\begin{align*}
\int_{\mathbb{S}^d \setminus C(\PT{x},\rho)} \left| \PT{x} - \PT{y} \right|^{-d} \dd \sigma_d(\PT{y}) 
&= \frac{\omega_{d-1}}{\omega_d} \int_{-1}^{1-\rho^2 / 2} \left[ 2 \left( 1 - t \right) \right]^{-d/2} \left( 1 - t^2 \right)^{d/2-1} \dd t \\
&= \frac{\omega_{d-1}}{\omega_d} 2^{-d/2} \int_{-1}^{1-\rho^2 / 2} \left( 1 - t \right)^{-1} \left( 1 + t \right)^{d/2-1} \dd t \\
&= \frac{1}{2} \frac{\omega_{d-1}}{\omega_d} \left( 1 - \rho^2/4 \right)^{d/2} \int_{0}^1 u^{d/2-1} \left( 1 - u \right)^{1-1} \left[ 1 - \left( 1 - \rho^2 / 4 \right) u \right]^{-1} \dd u \\
&= \frac{1}{2} \frac{\omega_{d-1}}{\omega_d} \left( 1 - \rho^2/4 \right)^{d/2} \gammafcn(d/2) \gammafcn(1) \HypergeomReg{2}{1}{1,d/2}{1+d/2}{1 - \rho^2 / 4}.
\end{align*}
The linear transformation \cite[Eq.~15.8.10]{DLMF2010.05.10} applied to the hypergeometric function 
above gives
\begin{align*}
&\int_{\mathbb{S}^d \setminus C(\PT{x},\rho)} \left| \PT{x} - \PT{y} \right|^{-d} \dd \sigma_d(\PT{y}) \\
&= - \frac{1}{2} \frac{\omega_{d-1}}{\omega_d} \left( 1 - \rho^2/4 \right)^{d/2} \sum_{k=0}^\infty \frac{\Pochhsymb{1}{k} \Pochhsymb{d/2}{k}}{k! k!} \left( \rho / 2 \right)^{2k} \left[ 2 \log( \rho / 2 ) + \digammafcn(k + d/2) - \digammafcn(k+1) \right] \\
&= - \frac{1}{2} \frac{\omega_{d-1}}{\omega_d} \left( 1 - \rho^2/4 \right)^{d/2} \Bigg\{ 2 \left( \log \frac{\rho}{2} \right) \sum_{k=0}^\infty \frac{\Pochhsymb{d/2}{k}}{k!} \left( \frac{\rho}{2} \right)^{2k} + \sum_{k=0}^\infty \frac{\Pochhsymb{d/2}{k}}{k!} \left[ \digammafcn(k + d/2) - \digammafcn(k+1) \right] \left( \frac{\rho}{2} \right)^{2k} \Bigg\} \\
&= \frac{\omega_{d-1}}{\omega_d} \left( - \log \rho \right) + \frac{\omega_{d-1}}{\omega_d} \log 2 - \frac{1}{2} \frac{\omega_{d-1}}{\omega_d} \left( 1 - \rho^2/4 \right)^{d/2} \sum_{k=0}^\infty \frac{\Pochhsymb{d/2}{k}}{k!} \left[ \digammafcn(k + d/2) - \digammafcn(k+1) \right] \left( \frac{\rho}{2} \right)^{2k}.
\end{align*}
The binomial expansion of the factor $( 1 - \rho^2/4 )^{d/2}$ is absolutely   convergent for   $0 \leq \rho <2$ as is the infinite series above. This gives
\begin{align*}
\left( 1 - \rho^2/4 \right)^{d/2} &\sum_{k=0}^\infty \frac{\Pochhsymb{d/2}{k}}{k!} \left[ \digammafcn(k + d/2) - \digammafcn(k+1) \right] \left( \frac{\rho}{2} \right)^{2k} \\
&= \sum_{k=0}^\infty \sum_{n=0}^\infty \frac{\Pochhsymb{d/2}{k}}{k!} \left[ \digammafcn(k + d/2) - \digammafcn(k+1) \right] \frac{\Pochhsymb{-d/2}{n}}{n!} \left( \frac{\rho}{2} \right)^{2k + 2n} = \sum_{m=0}^\infty b_m(d) \left( \frac{\rho}{2} \right)^{2 m}.
\end{align*}
where
\begin{equation*}  
b_m(d) := \sum_{k=0}^m \frac{\Pochhsymb{d/2}{k}\Pochhsymb{-d/2}{m-k}}{k!(m-k)!} \left[ \digammafcn(k + d/2) - \digammafcn(k+1) \right], 
\end{equation*}
Then $b_0(d) = \digammafcn(d/2) - \digammafcn(1)$ and, by Lemma~\ref{lem:identity}, 
\begin{equation*} \label{eq:identity}
b_m(d) =\frac{\Pochhsymb{1-d/2}{m}}{m! \, m} , \qquad m \geq 1.
\end{equation*}

Thus, we obtain
\begin{align*}
\int_{\mathbb{S}^d \setminus C(\PT{x},\rho)} \left| \PT{x} - \PT{y} \right|^{-d} \dd \sigma_d(\PT{y}) 
&= \frac{\omega_{d-1}}{\omega_d} \left( - \log \rho \right) + \frac{\omega_{d-1}}{\omega_d} \log 2 \\
&- \frac{1}{2} \frac{\omega_{d-1}}{\omega_d} \left\{ \digammafcn(d/2) - \digammafcn(1) + \sum_{m=1}^\infty \frac{\Pochhsymb{1-d/2}{m}}{m! \, m} \left( \frac{\rho}{2} \right)^{2 m} \right\}.
\end{align*}
\end{proof}

\begin{proof}[Proof of upper bound in Proposition~\ref{thm:2nd.term.boundary}]
A closer inspection of the second part of the proof of Theorem~3 in \cite{KuSa1998} gives almost (up to a $\log \log N$ factor) the correct order of the second term. Let $X_N^* = \{ \PT{x}_1^*, \dots, \PT{x}_N^* \}$ be an optimal $d$-energy configuration on $\mathbb{S}^d$. For $r > 0$ consider
\begin{equation*}
D_k(r) \DEF \mathbb{S}^d \setminus C(\PT{x}_k^*, r \, N^{-1/d}), \quad k = 1, \dots, N, \qquad D(r) \DEF \bigcap_{k=1}^N D_k(r).
\end{equation*}
Since $X_N^*$ is a minimal $d$-energy configuration, for each $j = 1, \dots, N$ the function
\begin{equation*}
U_j( \PT{x} ) \DEF \sum_{k: k \neq j} \frac{1}{\left| \PT{x}_k^* - \PT{x} \right|^d}, \qquad \PT{x} \in \mathbb{S}^d \setminus (X_N^*\setminus \{\PT{x}_j^*\}),
\end{equation*}
attains its minimum at   $\PT{x}_j^*$. By Lemma~\ref{lem:aux.2} and for $\rho = r \, N^{-1/d}$ ($0 < \rho <2$) we get
\begin{align*}
\sigma_d( D(r) ) \, U_j( \PT{x}_j ) 
&\leq \int_{D(r)} U_j( \PT{x} ) \dd \sigma_d( \PT{x} ) \leq \sum_{k: k \neq j} \int_{D_k(r)} \frac{1}{\left| \PT{x}_k^* - \PT{x} \right|^d} \dd \sigma_d( \PT{x} ) \\
&= \left( N - 1 \right) \Bigg\{ \frac{\omega_{d-1}}{\omega_d} \left( - \log(r N^{-1/d}) \right) - \frac{1}{2} \frac{\omega_{d-1}}{\omega_d} \left[ \digammafcn(d/2) - \digammafcn(1) - 2 \log 2 \right] \\
&\phantom{=}- \frac{1}{2} \frac{\omega_{d-1}}{\omega_d} \sum_{m=1}^\infty \frac{\Pochhsymb{1-d/2}{m}}{m!} \left( \frac{r N^{-1/d}}{2} \right)^{2 m} \Bigg\}.
\end{align*}
Hence,
\begin{equation*}
\begin{split}
\mathcal{E}_d( \mathbb{S}^d; N ) = \sum_{j=1}^N U_j( \PT{x}_j^*) 
&\leq \frac{N \left( N - 1 \right)}{\sigma_d( D(r) )} \Bigg\{ \frac{1}{d} \frac{\omega_{d-1}}{\omega_d} \log N +  \frac{\omega_{d-1}}{\omega_d} \left( - \log r \right) - \frac{1}{2} \frac{\omega_{d-1}}{\omega_d} \left[ \digammafcn(d/2) - \digammafcn(1) - 2 \log 2 \right] \\
&\phantom{=}- \frac{1}{2} \frac{\omega_{d-1}}{\omega_d} \sum_{m=1}^\infty \frac{\Pochhsymb{-d/2}{m}}{(m-1)!} \left( \frac{1}{m} - \frac{2}{d} \right) \left( \frac{r N^{-1/d}}{2} \right)^{2 m} \Bigg\}.
\end{split}
\end{equation*}
Subtracting off the dominant term of the asymptotic expansion, we obtain
\begin{equation*}
\begin{split}
&\mathcal{E}_d( \mathbb{S}^d; N ) - \frac{1}{d} \frac{\omega_{d-1}}{\omega_d} N^2 \log N \leq \frac{1}{d} \frac{\omega_{d-1}}{\omega_d} \frac{1 - \sigma_d( D(r) ) - 1 / N}{\sigma_d( D(r) )} N^2 \log N + \frac{1}{d} \frac{\omega_{d-1}}{\omega_d} \frac{N \left( N - 1 \right)}{\sigma_d( D(r) )} \left( - \log r^d \right) \\
&\phantom{=}- \frac{1}{2} \frac{\omega_{d-1}}{\omega_d} \left[ \digammafcn(d/2) - \digammafcn(1) - 2 \log 2 \right] \frac{N \left( N - 1 \right)}{\sigma_d( D(r) )} - \frac{1}{2} \frac{\omega_{d-1}}{\omega_d} \frac{N \left( N - 1 \right)}{\sigma_d( D(r) )}  \sum_{m=1}^\infty \frac{\Pochhsymb{-d/2}{m}}{(m-1)!} \left( \frac{1}{m} - \frac{2}{d} \right) \left( \frac{r N^{-1/d}}{2} \right)^{2 m}.
\end{split}
\end{equation*}
Using Lemma~\ref{lem:aux.1} in the trivial bound (recall $\rho = r N^{-1/d}$)
\begin{equation} \label{eq:trivial.bound}
\sigma_d( D(r) ) \geq 1 - N \sigma_d( C(\PT{x}_1^*, r N^{-1/d}) ) = 1 - \frac{1}{d} \frac{\omega_{d-1}}{\omega_d} \, r^d \, \Hypergeom{2}{1}{1-d/2,d/2}{1+d/2}{r^2 N^{-2/d} / 4}
\end{equation}
gives
\begin{equation*}
\left[ 1 - \sigma_d( D(r) ) \right] \log N \le \frac{1}{d} \frac{\omega_{d-1}}{\omega_d} \, r^d \log N \, \Hypergeom{2}{1}{1-d/2,d/2}{1+d/2}{r^2 N^{-2/d} / 4}.
\end{equation*}
Choosing $r^d = 1 / \log N$, we arrive at the result.
% \begin{equation*}
% \mathcal{E}_d( \mathbb{S}^d; N ) - \frac{1}{d} \frac{\omega_{d-1}}{\omega_d} N^2 \log N \leq 
% \end{equation*}
% 
\end{proof}

\begin{proof}[Proof of Proposition~\ref{thm:hypersing.lower.bound}]
We follow the Proof of Proposition~\ref{thm:2nd.term.boundary}, now for the kernel $K_\eps(s; t) := (2 - 2 t + \eps )^{-s/2}$ which is positive definite in the sense of Schoenberg~\cite{Sch1942} with the expansion $K_\eps(s;t) = \sum_{n = 0}^\infty a_n(s;\eps) \, P_n^d(t)$. (The positivity of the coefficients $a_n(s; \eps)$ can be seen by applying Rodrigues formula (see \cite{Mu1966}) and integration by parts $n$ times.) We have
\begin{equation*} 
E_s( X_N ) \geq a_0(s; \eps) \, N^2 - K_\eps(s; 1) \, N = a_0(s; \eps) \, N^2 - \eps^{-s/2} \, N,
\end{equation*}
where the coefficient $a_0(s; \eps)$ can be expressed in terms of a regularized Gauss hypergeomtric function 
\begin{align*}
a_0(s; \eps) 
&= \frac{\omega_{d-1}}{\omega_d} \int_{-1}^1 \left( 2 - 2 t + \eps \right)^{-s/2} \left( 1 - t^2 \right)^{d/2-1} \dd t \\
&= 2^{d-s-1} \frac{\omega_{d-1}}{\omega_d} \left( 1 + \eps / 4 \right)^{-s/2} \int_0^1 u^{d/2-1} \left( 1 - u \right)^{d/2-1} \left( 1 - \frac{1}{1+\eps/4} \, u \right)^{-s/2} \dd u \\
&= 2^{d-s-1} \frac{\omega_{d-1}}{\omega_d} \gammafcn(d/2) \gammafcn(d/2) \left( 1 + \eps / 4 \right)^{-s/2} \HypergeomReg{2}{1}{s/2,d/2}{d}{\frac{1}{1+\eps/4}}.
\end{align*}
For $(s-d)/2$ not an integer we can do asymptotic analysis by applying the linear transformation \cite[Eq.~15.8.5]{DLMF2010.05.10}, that is
\begin{equation*}
\begin{split}
\HypergeomReg{2}{1}{s/2,d/2}{d}{\frac{1}{1+\eps/4}} 
&= \frac{\pi}{\sin[\pi (d-s)/2]} \Bigg\{ \frac{\left(1 + \eps/4 \right)^{s/2}}{\gammafcn(d-s/2)\gammafcn(d/2)} \HypergeomReg{2}{1}{s/2,1-d+s/2}{1+(s-d)/2}{- \frac{\eps}{4}} \\
&\phantom{=}- \frac{\left[ \left( \eps / 4 \right) \big/ \left( 1 + \eps / 4 \right) \right]^{(d-s)/2} \left( 1 + \eps / 4 \right)^{d-s/2}}{\gammafcn(s/2) \gammafcn(d/2)} \HypergeomReg{2}{1}{d-s/2,1-s/2}{1-(s-d)/2}{- \frac{\eps}{4}} \Bigg\},
\end{split}
\end{equation*}
which, after simplifications and application of the last transformation in \cite[Eq.s~15.8.1]{DLMF2010.05.10}, yields
\begin{align*}
a_0(s; \eps) 
&= V_s(\mathbb{S}^d) \frac{1}{\gammafcn((d-s)/2)} \frac{\pi}{\sin[\pi (d-s)/2]} \HypergeomReg{2}{1}{s/2,1-d+s/2}{1+(s-d)/2}{- \frac{\eps}{4}} \\
&\phantom{=}- 2^{d-1-s} \frac{\omega_{d-1}}{\omega_d} \frac{\gammafcn(d/2)}{\gammafcn(s/2)} \frac{\pi}{\sin[\pi (d-s)/2]} \left( \frac{\eps}{4} \right)^{(d-s)/2} \HypergeomReg{2}{1}{1-d/2,d/2}{1-(s-d)/2}{- \frac{\eps}{4}}.
\end{align*}
Changing to Gauss hypergeometric functions shows that
\begin{equation*}
\begin{split}
a_0(s; \eps) 
&= 2^{d-1-s} \frac{\omega_{d-1}}{\omega_d} \frac{\gammafcn(d/2)}{\gammafcn(s/2)} \frac{\pi}{\sin[\pi (s-d)/2]} \frac{1}{\gammafcn(1-(s-d)/2)} \left( \frac{\eps}{4} \right)^{(d-s)/2} \Hypergeom{2}{1}{1-d/2,d/2}{1-(s-d)/2}{- \frac{\eps}{4}} \\
&\phantom{=}+ V_s(\mathbb{S}^d) \frac{1}{\gammafcn((d-s)/2)} \frac{1}{\gammafcn(1-(d-s)/2)} \frac{\pi}{\sin[\pi (d-s)/2]} \Hypergeom{2}{1}{s/2,1-d+s/2}{1+(s-d)/2}{- \frac{\eps}{4}}.
\end{split}
\end{equation*}
Application of the reflection formula for the gamma function \cite[Eq.s~5.5.3]{DLMF2010.05.10} and the substitution $\eps / 4 = a^{2/(d-s)} N^{-2/d}$ gives
\begin{equation}
\begin{split} \label{eq:a.0.s.eps}
a_0(s; \eps) 
&= 2^{d-1-s} \frac{\omega_{d-1}}{\omega_d} \frac{\gammafcn(d/2) \gammafcn((s-d)/2)}{\gammafcn(s/2)} \, a \, N^{-1+s/d} \Hypergeom{2}{1}{1-d/2,d/2}{1-(s-d)/2}{- a^{2/(d-s)} N^{-2/d}} \\
&\phantom{=}+ V_s(\mathbb{S}^d) \Hypergeom{2}{1}{s/2,1-d+s/2}{1+(s-d)/2}{- a^{2/(d-s)} N^{-2/d}}.
\end{split}
\end{equation}
Note that the first hypergeometric function above is a polynomial if $d$ is even and reduces to $1$ if $d=2$. Hence, using the series expansion of a hypergeometric function,
\begin{equation*}
\mathcal{E}_s( \mathbb{S}^d; N ) \geq f(s,d;a) \, N^{1+s/d} + V_s( \mathbb{S}^d ) \, N^2 + \mathcal{O}(N^{1+s/d-2/d}) \qquad \text{as $N \to \infty$,}
\end{equation*}
where the function $f(s,d;a) \DEF c \, a - 2^{-s} a^{s/(s-d)}$ with (cf. \eqref{eq:gammad.const})
\begin{equation*}
c = c_{s,d} = 2^{d-1-s} \frac{\omega_{d-1}}{\omega_d} \frac{\gammafcn(d/2) \gammafcn((s-d)/2)}{\gammafcn(s/2)} = 2^{d-1-s} \frac{\gammafcn((d+1)/2) \gammafcn((s-d)/2)}{\sqrt{\pi}\gammafcn(s/2)}
\end{equation*}
has a unique maximum at $a^* = [ 2^s c (s - d) / s ]^{(s-d)/d}$ with value 
\begin{align*}
A_{s,d} 
&\DEF f(s,d;a^*) = c \left[ 2^s c \left( s - d \right) / s \right]^{s/d-1} - 2^{-s} \left[ 2^s c \left( s - d \right) / s \right]^{s/d} = \frac{d}{s-d} \left[ 2^{s-d} \frac{s-d}{s} c \right]^{s/d} \\
&= \frac{d}{s-d} \left[ \frac{1}{2} \frac{\gammafcn((d+1)/2) \gammafcn(1+(s-d)/2)}{\sqrt{\pi}\gammafcn(1+s/2)} \right]^{s/d}.
\end{align*}
\end{proof}

For the proof of Proposition~\ref{thm:hypersing.upper.bound} we need the following auxiliary result.

\begin{lem} \label{lem:aux.3}
Let $d \geq 1$ and $s > d$ and $(s - d)/2$ not an integer. For $\PT{x} \in \mathbb{S}^d$ and $0 < \rho < 2$ we have
\begin{equation*}
\int_{\mathbb{S}^d \setminus C(\PT{x},\rho)} \left| \PT{x} - \PT{y} \right|^{-s} \dd \sigma_d(\PT{y}) = V_s(\mathbb{S}^d) + \frac{2^{d-s}}{s-d} \frac{\omega_{d-1}}{\omega_d} \left( \rho / 2 \right)^{d-s}  \Hypergeom{2}{1}{1-d/2,(d-s)/2}{1-(s-d)/2}{\frac{\rho^2}{4}}.
\end{equation*}
\end{lem}

\begin{proof}
Similar as in the proof of Lemma~\ref{lem:aux.2} we obtain
\begin{align*}
&\int_{\mathbb{S}^d \setminus C(\PT{x},\rho)} \left| \PT{x} - \PT{y} \right|^{-s} \dd \sigma_d(\PT{y}) = \frac{\omega_{d-1}}{\omega_d} \int_{-1}^{1-\rho^2 / 2} \left[ 2 \left( 1 - t \right) \right]^{-s/2} \left( 1 - t^2 \right)^{d/2-1} \dd t \\
&\phantom{equals}= \frac{\omega_{d-1}}{\omega_d} 2^{-s/2} \int_{-1}^{1-\rho^2 / 2} \left( 1 - t \right)^{(d-s)/2-1} \left( 1 + t \right)^{d/2-1} \dd t \\
&\phantom{equals}= 2^{d-1-s} \frac{\omega_{d-1}}{\omega_d} \left( 1 - \rho^2/4 \right)^{d/2} \int_{0}^1 u^{d/2-1} \left( 1 - u \right)^{1-1} \left[ 1 - \left( 1 - \rho^2 / 4 \right) u \right]^{(d-s)/2-1} \dd u \\
&\phantom{equals}= 2^{d-1-s} \frac{\omega_{d-1}}{\omega_d} \left( 1 - \rho^2/4 \right)^{d/2} \gammafcn(d/2) \gammafcn(1) \HypergeomReg{2}{1}{1+(s-d)/2,d/2}{1+d/2}{1 - \rho^2 / 4},
\intertext{for $(s-d)/2$ not a positive integer we can apply the linear transformation \cite[Eq.~15.8.4]{DLMF2010.05.10},}
% &\phantom{equals}= 2^{d-1-s} \frac{\omega_{d-1}}{\omega_d} \left( 1 - \rho^2/4 \right)^{d/2} \gammafcn(d/2) \gammafcn(1) \frac{\pi}{\sin[\pi(d-s)/2]}\Bigg\{ \frac{1}{\gammafcn(d-s/2) \gammafcn(1)} \HypergeomReg{2}{1}{1+(s-d)/2,d/2}{1+(s-d)/2}{\frac{\rho^2}{4}} \\
% &\phantom{equals=}- \frac{1}{\gammafcn(1+(s-d)/2) \gammafcn(d/2)} \left( \rho^2 / 4 \right)^{(d-s)/2} \HypergeomReg{2}{1}{d-s/2,1}{1-(s-d)/2}{\frac{\rho^2}{4}} \Bigg\}
&\phantom{equals}= 2^{d-1-s} \frac{\omega_{d-1}}{\omega_d} \left( 1 - \rho^2/4 \right)^{d/2} \frac{\pi}{\sin[\pi(d-s)/2]} \Bigg\{ \frac{\gammafcn(d/2) \gammafcn(1) }{\gammafcn(d-s/2) \gammafcn(1)} \HypergeomReg{2}{1}{1+(s-d)/2,d/2}{1+(s-d)/2}{\frac{\rho^2}{4}} \\
&\phantom{equals=}- \frac{\gammafcn(d/2) \gammafcn(1) }{\gammafcn(1+(s-d)/2) \gammafcn(d/2)} \left( \rho^2 / 4 \right)^{(d-s)/2} \HypergeomReg{2}{1}{d-s/2,1}{1-(s-d)/2}{\frac{\rho^2}{4}} \Bigg\},
\intertext{where the first regularized hypergeometric function can be evaluated using \cite[Eq.~15.4.6]{DLMF2010.05.10} and the second can be transformed by the last linear transformation in \cite[Eq.~15.8.1]{DLMF2010.05.10},}
&\phantom{equals}= 2^{d-1-s} \frac{\omega_{d-1}}{\omega_d} \left( 1 - \rho^2/4 \right)^{d/2} \frac{\pi}{\sin[\pi(d-s)/2]} \Bigg\{ \frac{\gammafcn(d/2)}{\gammafcn(d-s/2) \gammafcn(1+(s-d)/2)} \left( 1 - \rho^2 / 4 \right)^{-d/2} \\
&\phantom{equals=}- \frac{1}{\gammafcn(1+(s-d)/2)} \left( \rho^2 / 4 \right)^{(d-s)/2}  \left( 1 - \rho^2 / 4 \right)^{-d/2} \HypergeomReg{2}{1}{1-d/2,(d-s)/2}{1-(s-d)/2}{\frac{\rho^2}{4}} \Bigg\} \\
&\phantom{equals}= 2^{d-1-s} \frac{\omega_{d-1}}{\omega_d} \frac{\gammafcn(d/2) }{\gammafcn(d-s/2)} \frac{\pi}{\sin[\pi(d-s)/2]} \frac{1}{ \gammafcn(1-(d-s)/2)} \\
&\phantom{equals=}- 2^{d-1-s} \frac{\omega_{d-1}}{\omega_d} \frac{\pi}{\sin[\pi(d-s)/2]} \frac{1}{\gammafcn(1+(s-d)/2)} \left( \rho / 2 \right)^{d-s}  \HypergeomReg{2}{1}{1-d/2,(d-s)/2}{1-(s-d)/2}{\frac{\rho^2}{4}},
\intertext{after changing to a non-regularized hypergeometric function and using the reflection formula for the gamma function,}
% &\phantom{equals}= 2^{d-1-s} \frac{\omega_{d-1}}{\omega_d} \frac{\gammafcn(d/2) \gammafcn((d-s)/2)}{\gammafcn(d-s/2)}  \\
% &\phantom{equals=}- 2^{d-1-s} \frac{\omega_{d-1}}{\omega_d} \frac{\pi}{\sin[\pi(d-s)/2]} \frac{1}{\gammafcn(1+(s-d)/2)\gammafcn(1-(s-d)/2)} \left( \rho / 2 \right)^{d-s}  \Hypergeom{2}{1}{1-d/2,(d-s)/2}{1-(s-d)/2}{\frac{\rho^2}{4}},
&\phantom{equals}= 2^{d-1-s} \frac{\omega_{d-1}}{\omega_d} \frac{\gammafcn(d/2) \gammafcn((d-s)/2)}{\gammafcn(d-s/2)} + \frac{2^{d-s}}{s-d} \frac{\omega_{d-1}}{\omega_d} \left( \rho / 2 \right)^{d-s}  \Hypergeom{2}{1}{1-d/2,(d-s)/2}{1-(s-d)/2}{\frac{\rho^2}{4}}.
\end{align*}
% For $(s-d)/2$ not a positive integer we can apply the linear transformation \cite[Eq.~15.8.4]{DLMF2010.05.10}
The substitution $\omega_{d-1}/\omega_d = \gammafcn((d+1)/2) / [ \sqrt{\pi} \gammafcn(d/2) ]$ (cf. \eqref{eq:gammad.const}) shows that the first term is $V_s(\mathbb{S}^d)$ by Eq.~\eqref{V.s.d}. The result follows.
\end{proof}

\begin{proof}[Proof of Proposition~\ref{thm:hypersing.upper.bound}]
We follow the proof of the upper bound in Proposition~\ref{thm:2nd.term.boundary}. Since $X_N^*$ is a minimal $s$-energy configuration, for each $j = 1, \dots, N$ the function
\begin{equation*}
U_j( \PT{x} ) \DEF \sum_{k: k \neq j} \frac{1}{\left| \PT{x}_k^* - \PT{x} \right|^s}, \qquad \PT{x} \in \mathbb{S}^d \setminus \left(X_N^*\setminus\{\PT{x}_j^*\}\right),
\end{equation*}
attains its minimum at $\PT{x}_j^*$. By Lemma~\ref{lem:aux.3} with $\rho = r \, N^{-1/d}$ ($0 < \rho  <2$) we get
\begin{align*}
\sigma_d( D(r) ) \, U_j( \PT{x}_j ) 
&\leq \int_{D(r)} U_j( \PT{x} ) \dd \sigma_d( \PT{x} ) \leq \sum_{k: k \neq j} \int_{D_k(r)} \frac{1}{\left| \PT{x}_k^* - \PT{x} \right|^s} \dd \sigma_d( \PT{x} ) \\
&= \left( N - 1 \right) \Bigg\{ V_s(\mathbb{S}^d) + \frac{1}{s-d} \frac{\omega_{d-1}}{\omega_d} r^{d-s} N^{-1+s/d}  \Hypergeom{2}{1}{1-d/2,(d-s)/2}{1-(s-d)/2}{\frac{r^2}{4} N^{-2/d}} \Bigg\}.
\end{align*}
Hence
\begin{align*}
\sigma_d(D(r)) \, \mathcal{E}_s( \mathbb{S}^d; N) 
&\leq N^2 \Bigg\{ V_s(\mathbb{S}^d) + \frac{1}{s-d} \frac{\omega_{d-1}}{\omega_d} r^{d-s} N^{-1+s/d}  \Hypergeom{2}{1}{1-d/2,(d-s)/2}{1-(s-d)/2}{\frac{r^2}{4} N^{-2/d}} \Bigg\} \\
&= \frac{1}{s-d} \frac{\omega_{d-1}}{\omega_d} r^{d-s} N^{1+s/d} + V_s(\mathbb{S}^d) \, N^2 + \mathcal{O}(r^{2+d-s} \, N^{1 + s / d - 2 / d} ).
\end{align*}
By   relations~\eqref{EsSd} and \eqref{eq:trivial.bound}, we have
\begin{equation*}
\sigma_d( D(r) ) \, \mathcal{E}_s( \mathbb{S}^d; N) \geq \left( 1 - \frac{1}{d} \frac{\omega_{d-1}}{\omega_d} \, r^d \right) \mathcal{E}_s( \mathbb{S}^d; N) + \mathcal{O}(r^{2+d} \, N^{1 + s / d - 2 / d} ).
\end{equation*}
Note, that for $d=2$ the hypergeometric function in \eqref{eq:trivial.bound} reduces to one and the $\mathcal{O}(.)$-term above disappears. Hence
\begin{equation*}
\mathcal{E}_s( \mathbb{S}^d; N) \leq \frac{1}{s-d} \frac{\omega_{d-1}}{\omega_d} \frac{r^{d-s}}{1 - \frac{1}{d}  \frac{\omega_{d-1}}{\omega_d} \, r^d} \, N^{1+s/d} + \frac{V_s(\mathbb{S}^d)}{1 - \frac{1}{d} \frac{\omega_{d-1}}{\omega_d} \, r^d} \, N^2 + \mathcal{O}(r^{2+d-s} \, N^{1 + s / d - 2 / d} ).
\end{equation*}
The function $h(r) = r^{d-s} / (1 - c r^d)$ (where $c = (1/d) (\omega_{d-1}/\omega_d)$) has a single minimum in the interval $(0,\infty)$ at $r^* = c^{-1/d} ( 1 - d / s)^{1/d}$ with value $h(r^*) = (s/d) c^{-1+s/d} ( 1 - d / s)^{1-s/d}$, where $1 - c r^d = d/s > 0$. The result follows.
\end{proof}

\section{Motivations for conjectures}
\label{sec:justification}

\begin{proof}[Motivation for Conjecture~\ref{conj:log.2.sphere}]
Suppose Conjecture~\ref{conj:Riesz.d.sphere} holds. Proceeding formally, we obtain
\begin{equation*}
\mathcal{E}_{\mathrm{log}}( \mathbb{S}^d; N ) 
= \left. \frac{\dd}{\dd s} \mathcal{E}_{s}( \mathbb{S}^d; N ) \right|_{s\to0^+} =  \left. \frac{\dd}{\dd s} \left[V_s(\mathbb{S}^d) N^2 + \left(\frac{|\Lambda_d|}{ \omega_d} \right)^{s/d}\zetafcn_{\Lambda_d}(s) \, N^{1+s/d}  \right]\right|_{s\to0^+} + \Delta( \mathbb{S}^d; N ),
\end{equation*}
where 
\begin{equation*}
\Delta( \mathbb{S}^d; N ) = \lim_{s\to0^+} \frac{1}{s} \left[ \mathcal{E}_s( \mathbb{S}^d; N ) - V_s({\mathbb S}^d) \, N^2 - \left(\frac{|\Lambda_d|}{ \omega_d} \right)^{s/d}\zetafcn_{\Lambda_d}(s) \, N^{1+s/d} \right].
\end{equation*}
Assuming   the limit exists and that $\Delta( \mathbb{S}^d; N ) = o(N)$ as $N \to \infty$, we  have
\begin{equation}
\begin{split}
\mathcal{E}_{\mathrm{log}}( \mathbb{S}^d; N ) &= \left. \frac{\dd}{\dd s} \left [ V_s(\mathbb{S}^d) N^2 + \left(\frac{|\Lambda_d|}{ \omega_d} \right)^{s/d}\zetafcn_{\Lambda_d}(s) \, N^{1+s/d}  \right ] \right|_{s\to0^+} + o(N), \\
&=  V_{\log}(\mathbb{S}^d) N^2 +\zetafcn_{\Lambda_d}(0)\frac{1}{d} \, N\log N   + 
\left[ \zetafcn_{\Lambda_d}'(0)+\frac{\zetafcn_{\Lambda_d}(0)}{d}\log \left(\frac{|\Lambda_d|}{ \omega_d} \right) \right]N + o(N) ,\\
&=  V_{\log}(\mathbb{S}^d) N^2 -\frac{1}{d} \, N\log N   + 
\left[ \zetafcn_{\Lambda_d}'(0)-\frac{1}{d}\log \left(\frac{|\Lambda_d|}{ \omega_d} \right) \right]N + o(N),\\
&=  V_{\log}(\mathbb{S}^d) N^2 -\frac{1}{d} \, N\log N   + 
C_{{\log }, d}N + o(N) \quad \text{as $N \to \infty$,}
\end{split}
\end{equation}
%(In light of Smale's Problem \#7 it is more likely that $\Delta( \mathbb{S}^2; N ) = \mathcal{O}(\log N)$.)
where (cf. \eqref{V.log} and \eqref{V.s.d}) we used
\begin{equation*}
\left. \frac{\dd}{\dd s} V_s(\mathbb{S}^d) \right|_{s\to0^+} = \left. \frac{\dd}{\dd s} \left [ 2^{d-s-1} \frac{\gammafcn((d+1)/2) \gammafcn((d-s)/2)}{\sqrt{\pi}\gammafcn(d-s/2)}\right ] \right|_{s\to0^+}=
\log \frac{1}{2} + \frac{1}{2} \left[ \digammafcn( d ) - \digammafcn( d / 2 ) \right]=V_{\log}(\mathbb{S}^d),\end{equation*}  and also used the fact that  $\zetafcn_{\Lambda}(0)=-1$ holds for any lattice $\Lambda$ (cf.  \cite{Te1988}).

 Using  Proposition~\ref{prop:zeta.lambda.at.0} in the appendix and $|\Lambda_2|=\sqrt{3}/2$ (for the hexagonal lattice with unit length edges), we  obtain
\begin{equation}
C_{\mathrm{log},2} 
= \zetafcn_{\Lambda_2}^\prime( 0 )-  \frac{1}{2} \log \frac{\sqrt{3}}{8 \pi}  =    \log ( 2 \pi ) - \frac{\log 3}{4} - 3 \log \gammafcn(1/3)- \frac{1}{2} \log \frac{\sqrt{3}}{8 \pi}
= 2 \log 2 + \frac{1}{2} \log \frac{2}{3} + 3 \log \frac{\sqrt{\pi}}{\gammafcn(1/3)}.
\end{equation}

\end{proof}

\begin{proof}[Motivation for  Conjecture~\ref{conj:Riesz.d.dsphere}]
We first remark that   $ V_s( \mathbb{S}^d )$ has a simple pole at $s=d$ with 
\begin{equation*}
V_s( \mathbb{S}^d )=\frac{a_{-1,d}}{s-d}+A_d+\mathcal{O}(|s-d|) \text{ as $s\to d$,}
\end{equation*}
where
\begin{equation*}
a_{-1,d} \DEF \res_{s=d} V_s( \mathbb{S}^d ) = - d \frac{\mathcal{H}_d(\mathbb{B}^d)}{\mathcal{H}_d(\mathbb{S}^d)}$$ and $$ A_d \DEF \lim_{s\to d} \left[ V_s( \mathbb{S}^d ) - \frac{a_{-1,d}}{s-d} \right]= - \frac{1}{2} \frac{\omega_{d-1}}{\omega_d} \left( \gamma - 2 \log 2 + \digammafcn( d / 2 ) \right).
\end{equation*}
In addition to  Conjecture~\ref{conj:Riesz.d.sphere}, we further assume that $C_{s,d}$ behaves, near $s=d$, as  follows
$$\frac{C_{s,d}}{[\mathcal{H}_d({\mathbb S}^d)]^{s/d}}=b_{-1,d}/(s-d)+B_d +\mathcal{O}(|s-d|) \text { as $s\to d$, }$$
 where $$b_{-1,d}:=-a_{-1,d}  \quad \text{   and  }\quad
 B_d \DEF \lim_{s\to d} \left[ \frac{C_{s,d}}{\left[ \mathcal{H}_d( \mathbb{S}^d ) \right]^{s/d}} - \frac{b_{-1,d}}{s-d} \right].$$
 
Proceeding similarly as before and taking $s\to d$, we have
\begin{align*}
\mathcal{E}_s( \mathbb{S}^d; N ) 
&= V_s( \mathbb{S}^d ) \, N^2 + \frac{C_{s,d}}{\left[ \mathcal{H}_d( \mathbb{S}^d ) \right]^{s/d}} \, N^{1+s/d} + \Delta_s( \mathbb{S}^d; N ) \\
&= \left[ V_s( \mathbb{S}^d ) - \frac{a_{-1,d}}{s-d} \right] N^2 + \frac{a_{-1,d}}{s-d} \left( N^2 - N^{1+s/d} \right) \\
&\phantom{=}+ \left[ \frac{C_{s,d}}{\left[ \mathcal{H}_d( \mathbb{S}^d ) \right]^{s/d}} - \frac{b_{-1,d}}{s-d} \right] N^{1+s/d} + \Delta_s( \mathbb{S}^d; N ) \\
& \longrightarrow  A_d \, N^2 + \frac{\mathcal{H}_d(\mathbb{B}^d)}{\mathcal{H}_d(\mathbb{S}^d)} \, N^2 \log N + B_d N^2 + \Delta_d( \mathbb{S}^d; N ),
\end{align*}
where we assume the limit $\Delta_d( \mathbb{S}^d; N )$ exists and that $\Delta_d( \mathbb{S}^d; N ) = o(N^2)$ as $N \to \infty$.

Multiplying both sides of \eqref{eq:C.s.d.expansion} with $\left[ \mathcal{H}_d( \mathbb{S}^d ) \right]^{s/d}$ and expanding $\left[ \mathcal{H}_d( \mathbb{S}^d ) \right]^{s/d}$ about $s = d$, we obtain that
\begin{equation}\label{eq:C.s.d.expansion}
C_{s,d} = \frac{\omega_{d-1}}{s-d} + \widetilde{B}_d + \mathcal{O}(|s-d|) \qquad \text{as $s\to d$,}
\end{equation}
where
\begin{equation*}
\widetilde{B}_d \DEF \lim_{s \to d} \left[ C_{s,d} - \frac{\omega_{d-1}}{s-d} \right].
\end{equation*}
Furthermore, the following relation holds between coefficient of the $N^2$-term and $\widetilde{B}_d$:
\begin{align*}
C_{d,d} 
&\DEF A_d + B_d = - \frac{1}{2} \frac{\omega_{d-1}}{\omega_d} \left( \gamma - 2 \log 2 + \digammafcn( d / 2 ) \right) - \frac{1}{d} \frac{\omega_{d-1}}{\omega_d} \log \omega_d + \frac{\widetilde{B}_d}{\omega_d}.
\end{align*}

In the case $d=2$ (with the help of Mathematica), we obtain 
\begin{equation*}
a_{-1,2} = \res_{s=2} \frac{2^{1-s}}{2-s} = - 1 / 2, \qquad b_{-1,2} = \res_{s=2} \frac{C_{s,2}}{\left( 4 \pi \right)^{s/2}} = 1 / 2, 
\end{equation*}
and
\begin{align*}
A_2=\lim_{s\to2} \left[ \frac{2^{1-s}}{2-s} - \frac{a_{-1,2}}{s-2} \right] &= \frac{\log 2}{2}, \\
B_2 = \lim_{s\to2} \left[ \frac{C_{s,2}}{\left( 4 \pi \right)^{s/2}} - \frac{b_{-1,2}}{s-2} \right] &= \frac{1}{4} \left[ \gamma - \log ( 8 \sqrt{3} \pi ) \right] + \frac{\sqrt{3}}{4 \pi} \left[ \gamma_1(2/3) - \gamma_1(1/3) \right],
\end{align*}
where $\gamma$ is the Euler-Mascheroni constant and $\gamma_n(a)$ is the generalized Stieltjes constant appearing as the coefficient of $(1-s)^n$ in the expansion of $\zetafcn(s,a)$ about $s=1$.

\end{proof}

\appendix

\section{Auxiliary results}

\begin{prop} \label{prop:zeta.lambda.at.0}
\begin{align*}
\DirichletL_{-3}(0) &= 1 / 3, & \DirichletL_{-3}^\prime(0) &= - \frac{1}{3} \log 3 + \log \frac{\gammafcn(1/3)}{\gammafcn(2/3)} = \frac{\log 3}{6} + 2 \log \frac{\gammafcn(1/3)}{\sqrt{2 \pi}}, \\
\zetafcn_{\Lambda_2}(0) &= -1, & \zetafcn_{\Lambda_2}^\prime(0) &= \log \frac{\sqrt{3}}{\sqrt{2 \pi}} + \frac{3}{2} \log \frac{\gammafcn(2/3)}{\gammafcn(1/3)} = \log ( 2 \pi ) - \frac{\log 3}{4} - 3 \log \gammafcn( 1 / 3 ).
\end{align*}
\end{prop}

\begin{proof}
Note that one has the following identities (cf. \cite[p.~264]{Ap1976})
\begin{equation} \label{app.ids}
\zetafcn(0,a) = (1 / 2) - a, \qquad \left. \frac{\dd }{\dd s} \zetafcn(s,a) \right|_{s\to0} = \log \gammafcn(a) - \log \sqrt{2 \pi}.
\end{equation}
By \eqref{dirichlet.id} and the last two relations
\begin{align*}
\DirichletL_{-3}(0) 
&= \zetafcn(0,1/3) - \zetafcn(0,2/3) = \left( 1 / 2 \right) - \left( 1 / 3 \right) - \left( 1 / 2 \right) + \left( 2 / 3 \right) = 1 / 3, \\
\DirichletL_{-3}^\prime(0) 
&= \left. \frac{\dd }{\dd s} \left\{ 3^{-s} \left[ \zetafcn(s,1/3) - \zetafcn(s,2/3) \right] \right\} \right|_{s\to0} = - \log 3 \, \DirichletL_{-3}(0) + \log \gammafcn(1/3) - \log \gammafcn(2/3).
\end{align*} 
For the second relation for $\DirichletL_{-3}^\prime(0)$ we used $\gammafcn(2/3) = 2 \pi / [ \sqrt{3} \gammafcn(1 / 3) ]$.

By \eqref{zeta.lambda.prod} and $\zetafcn(0) = - 1 / 2$, we get
\begin{align*}
\zetafcn_{\Lambda_2}(0) &= 6 \zetafcn( 0 ) \DirichletL_{-3}( 0 ) = 6 \left(- 1 / 2 \right) \left( 1 / 3 \right) = -1 \\
\zetafcn_{\Lambda_2}^\prime(0) &= \left. \frac{\dd}{\dd s}\left\{ 6 \zetafcn( s / 2 ) \DirichletL_{-3}( s / 2 ) \right\} \right|_{s \to 0} = 3 \zetafcn^\prime(0) \DirichletL_{-3}(0) + 3 \zetafcn(0) \DirichletL_{-3}^\prime(0).
\end{align*}
Since $\zetafcn(s) = \zetafcn(s,1)$, we derive from \eqref{app.ids} the special values $\zetafcn(0) = - 1 / 2$ and $\zetafcn^\prime(0) = - \log \sqrt{2 \pi}$. This completes the proof.
\end{proof}

\bibliographystyle{abbrv}
\bibliography{ENERGYbibliography}
%\bibliography{bibliography}

\end{document}